\documentclass[reqno,11pt]{amsart}

\usepackage[T1]{fontenc}
\usepackage{lmodern}
\usepackage{amsmath,amssymb,amsthm,mathtools}
\usepackage[margin=1in]{geometry}
\usepackage{enumitem}
\usepackage[protrusion=true,expansion=false]{microtype}
\usepackage[hidelinks]{hyperref}

\newcommand{\R}{\mathbb R}
\newcommand{\dd}{\,\mathrm d}
\newcommand{\Ric}{\operatorname{Ric}}
\newcommand{\BK}{\mathcal M_{\mathrm{BK}}}
\newcommand{\Mcan}{\mathcal M_{\mathrm{can}}}

\numberwithin{equation}{section}

\theoremstyle{plain}
\newtheorem{theorem}{Theorem}[section]
\newtheorem{proposition}[theorem]{Proposition}
\newtheorem{lemma}[theorem]{Lemma}
\newtheorem{corollary}[theorem]{Corollary}
\newtheorem{maintheorem}{Theorem}

\theoremstyle{definition}
\newtheorem{definition}[theorem]{Definition}

\theoremstyle{remark}
\newtheorem{remark}[theorem]{Remark}

\title[$C^0$-inextendibility of Birmingham-Kottler spacetimes]
{Warped Spacelike Singularities and the $C^0$-Inextendibility of Birmingham-Kottler Spacetimes}

\author{Bobby Eka Gunara}
\address{Theoretical Physics Laboratory, Theoretical High Energy Physics Research Division,
Faculty of Mathematics and Natural Sciences, Institut Teknologi Bandung,
Jl.\ Ganesha no.\ 10, Bandung 40132, Indonesia}
\email{bobby@itb.ac.id}

\date{}

\hypersetup{
  pdftitle={Warped Spacelike Singularities and the C0-Inextendibility of Canonical One-Horizon Birmingham-Kottler Spacetimes},
  pdfauthor={Bobby Eka Gunara},
  pdfsubject={Low-regularity Lorentzian geometry and black-hole spacetimes},
  pdfkeywords={C0-inextendibility, Birmingham-Kottler spacetime, warped product, spacelike singularity}
}

\begin{document}
\raggedbottom

\begin{abstract}
We establish a local obstruction to continuous Lorentzian extensions at a
class of warped spacelike singularities.  The criterion is expressed through
two integrability conditions, a monotonicity condition on the relative warp
factors, and divergence of the longitudinal factor; it does not require any
symmetry of the closed fiber.  The main geometric step is a radial compression
of terminal causal traces, which replaces the rotational deformation available
in spherical symmetry.  The compression yields a compact chronological
separator in an adapted boundary chart.  Longitudinal translations then
produce radial-slice distances that diverge intrinsically while remaining
uniformly controlled in the extension chart.  Next, we construct the canonical
one-horizon Birmingham-Kottler spacetime in global Kruskal coordinates and
classify all finite proper-time ends of its timelike geodesics.  Every
boundary-approaching finite maximizer supplied by a putative extension is
thereby forced to the singular end.  The local obstruction and the geodesic
classification imply $C^0$-inextendibility for the one-horizon
Birmingham-Kottler family with nonpositive cosmological constant and every
closed connected fiber satisfying
$\operatorname{Ric}_{\gamma_\Sigma}=(n-2)k\gamma_\Sigma$, without assumptions
of homogeneity, orientability, or simple connectivity.
\end{abstract}

\maketitle

\tableofcontents

\section{Introduction}\label{sec:introduction}

A smooth spacetime is $C^0$-inextendible if it cannot be embedded as a proper
open subset of a spacetime with a continuous Lorentzian metric.  At this
regularity curvature is unavailable: even the blow-up of a scalar invariant
does not, by itself, rule out a continuous extension.  Proofs of
$C^0$-inextendibility must instead use features of the metric and its causal
geometry that remain meaningful without derivatives.  Sbierski's treatment of
maximal Schwarzschild spacetime is the basic example
\cite{SbierskiSchw,SbierskiProof}.  Further criteria use causal divergence and
one-connectedness, or timelike geodesic completeness together with global
hyperbolicity \cite{GL,GLS}.

The purpose of this paper is to isolate a local obstruction at a warped
spacelike singularity and to combine it with a global analysis of the
one-horizon Birmingham-Kottler family
\cite{Kottler,Birmingham,Tangherlini}.  The local result is proved first and
then applied to the canonical global spacetime.

Throughout the paper, a closed manifold is understood to be compact and
without boundary.

\subsection{The local obstruction}

Let $(\Sigma,\gamma_\Sigma)$ be an arbitrary closed connected Riemannian
manifold, and consider
\begin{equation}\label{eq:intro-warped-metric}
 M_R=(0,R)_r\times\R_t\times\Sigma,
 \qquad
 g=-a(r)^2\dd r^2+b(r)^2\dd t^2+c(r)^2\gamma_\Sigma,
\end{equation}
where $a,b,c$ are smooth and positive.  We orient time so that $r$ decreases
along future-directed causal curves and put
\begin{equation}\label{eq:intro-AB}
 A(r)=\int_0^r\frac{a(s)}{b(s)}\dd s,
 \qquad
 B(r)=\int_0^r\frac{a(s)}{c(s)}\dd s.
\end{equation}

The first main result is the local warped-singularity obstruction.  It is
purely metric and uses no field equation.

\begin{maintheorem}[Warped spacelike singularities]
\label{thm:main-local}
Suppose that there is $R_0\in(0,R)$ for which the following conditions hold.
The values $A(R_0)$ and $B(R_0)$ are finite, the quotient $b/c$ is
nonincreasing on $(0,R_0)$, and
$b(r)\to\infty$ as $r\searrow0$.  No $C^0$ extension of
\eqref{eq:intro-warped-metric} contains a future boundary point approached by
a future-directed timelike curve along which $r\to0$.  The analogous statement
holds after reversing the time orientation.
\end{maintheorem}

The geometric meaning of the assumptions is visible in the causal estimate
\begin{equation}\label{eq:intro-causal-budget}
 \left(\frac{b}{a}\right)^2|t'|^2
 +\left(\frac{c}{a}\right)^2|\omega'|_{\gamma_\Sigma}^2\le1.
\end{equation}
The functions $A$ and $B$ control, respectively, the variation of $t$ and the
length of the fiber projection.  Both components therefore converge as
$r\searrow0$.  The divergence of $b$ supplies the final contradiction: on the
complete hypersurface $\{r=\rho\}$, points whose $t$-coordinates differ by
$2\lambda$ have distance at least $2\lambda b(\rho)$.

The new point is the localization of this divergence for a fiber without
isometries.  For $0<\varepsilon<r$ and $0<\vartheta<1$, define
$\phi_\varepsilon$ by
\begin{equation}\label{eq:intro-compression}
 B(\phi_\varepsilon(r))
 =\vartheta\bigl(B(r)-B(\varepsilon)\bigr).
\end{equation}
Differentiation gives
\[
 \phi_\varepsilon'(r)
 =\vartheta\frac{c(\phi_\varepsilon(r))/a(\phi_\varepsilon(r))}
                         {c(r)/a(r)}.
\]
The angular contribution to \eqref{eq:intro-causal-budget} is multiplied by
$\vartheta^2$.  The time contribution is multiplied by at most the same
factor, since $\phi_\varepsilon(r)\le r$ and $b/c$ is nonincreasing.  This
compression moves the limiting trace from $r=0$ to a positive radial level
through uniformly timelike curves.  It produces a compact separator in an
adapted boundary chart.  Translations in the $t$-direction then give two points
on each of a sequence of small radial levels whose chart distance is uniformly
bounded, whereas their intrinsic distance tends to infinity.

\subsection{The global application}

Let $n\ge3$, $m>0$, and $\Lambda\le0$, and set
\begin{equation}\label{eq:intro-f}
 f(r)=k-\frac{2m}{r^{n-2}}-
       \frac{2\Lambda}{n(n-1)}r^2,
 \qquad k\in\{-1,0,1\},
\end{equation}
where $k=1$ if $\Lambda=0$.  If the closed connected manifold
$(\Sigma^{n-1},\gamma_\Sigma)$ satisfies
\begin{equation}\label{eq:intro-Einstein-fiber}
 \Ric_{\gamma_\Sigma}=(n-2)k\gamma_\Sigma,
\end{equation}
then the block metric
\begin{equation}\label{eq:intro-BK-block}
 g=-f(r)\dd t^2+f(r)^{-1}\dd r^2+r^2\gamma_\Sigma
\end{equation}
solves $\Ric_g=2\Lambda g/(n-1)$.  The function $f$ is strictly increasing and
has one positive simple zero.  We denote by $\BK$ the canonical Kruskal
extension through this horizon.  The preceding local result is then applied
to obtain the second main theorem.

\begin{maintheorem}[Canonical one-horizon Birmingham-Kottler spacetimes]
\label{thm:main-global}
The spacetime $\BK$ is $C^0$-inextendible.  The conclusion holds for every
closed connected Einstein fiber satisfying
\eqref{eq:intro-Einstein-fiber}; the fiber need not be homogeneous, orientable,
or simply connected.
\end{maintheorem}

Completeness criteria do not apply to $\BK$, since timelike geodesics reach the
spacelike singularity in finite proper time.  Instead, a theorem of Minguzzi
and Suhr \cite{MinguzziSuhr} shows that a continuous extension would produce a
finite positive-length causal maximizer approaching its boundary.  In the
smooth part of the original spacetime this maximizer, parametrized by proper
time, is a timelike geodesic.

The Kruskal construction and the geodesic equations reduce the possible ends
of this curve to an explicit alternative.  The static Killing field and the
warped-product connection give constants $E\in\mathbb R$ and $J\ge0$ with
\begin{equation}\label{eq:intro-radial-energy}
 \dot r^2=E^2-f(r)\left(1+\frac{J^2}{r^2}\right).
\end{equation}
Both Kruskal null coordinates are strictly monotone along a timelike curve, so
its tail lies in one quadrant.  Equation \eqref{eq:intro-radial-energy} then
shows that a finite geodesic either converges to a point of the Kruskal
spacetime, possibly on the horizon, or enters a dynamic block and satisfies
$r\to0$.  A boundary-approaching maximizer leaves every compact subset, hence
only the second alternative can occur.  In the dynamic block, the coefficients
in Theorem~\ref{thm:main-local} are
\[
 a=(-f)^{-1/2},\qquad b=(-f)^{1/2},\qquad c=r.
\]
Their precise asymptotics and the required monotonicity are verified in
Section~\ref{sec:BK-global}, where Theorem~\ref{thm:main-local} completes the proof.

\subsection{Relation to earlier work}

Sbierski introduced the spacelike diameter and used it to exclude a continuous
extension of maximal Schwarzschild spacetime through $r=0$
\cite{SbierskiSchw}.  The argument below has the same metric origin, but the
radial compression \eqref{eq:intro-compression} replaces the rotations used in
the spherically symmetric setting.  Mosani has recently treated a broad class
of warped-product black holes, including Birmingham metrics at the central
singularity, under assumptions that include a homogeneous orientable fiber
\cite{Mosani}.  Here the local theorem requires neither hypothesis.  The
one-horizon Kruskal and geodesic analysis also yields inextendibility of the
entire canonical spacetime, rather than only the exclusion of a boundary point
at the central singularity.

The global step is different from the timelike-completeness obstruction of
Galloway, Ling, and Sbierski \cite{GLS}.  Radial timelike geodesics in the
present spacetime are incomplete.  What is used instead is that every finite
timelike geodesic escaping compact subsets must reach the singular end to
which the local theorem applies.

\subsection{Organization of the paper}

In Section~\ref{sec:boundary} we collect the low-regularity boundary results and
extracts a finite timelike geodesic from a putative extension. Then, we prove Theorem~\ref{thm:main-local} in
Section~\ref{cen:sec:center}.  The global Kruskal
construction, the Einstein equation, the geodesic first integrals, and the
proof of Theorem~\ref{thm:main-global} are given in Section~\ref{sec:BK-global}. The final
section records the conclusions and some directions not covered by the
one-horizon argument.

\section{Continuous extensions and finite maximizing geodesics}
\label{sec:boundary}

\subsection{Definitions and local coordinates}

\begin{definition}[$C^0$ extension]\label{def:C0-spacetime}
Let $(M,g)$ be a smooth, connected, time-oriented spacetime.  A
\emph{$C^0$ extension} of $(M,g)$ is a connected smooth manifold
$\widetilde M$ of the same dimension, equipped with a continuous Lorentzian
metric $\widetilde g$ and a continuous time orientation, together with a
time-orientation-preserving $C^1$ embedding
\[
 \iota:M\longrightarrow \widetilde M
\]
such that
\[
 \iota^*\widetilde g=g,
 \qquad
 \iota(M)\ne\widetilde M.
\]
Here a continuous time orientation means a continuous timelike vector field;
the future cone is the component containing that field.  Since $\iota$ is an
embedding between manifolds of the same dimension, $\iota(M)$ is open.  It is
a nonempty proper open subset of the connected manifold $\widetilde M$, and
therefore
\[
 \partial\iota(M)=\overline{\iota(M)}\setminus\iota(M)\ne\varnothing.
\]
The spacetime $(M,g)$ is \emph{$C^0$-inextendible} if it has no such
extension.
\end{definition}

The $C^1$ regularity of the embedding in
Definition~\ref{def:C0-spacetime} is the standard choice in the
low-regularity extension problem and is precisely what is needed to define
the pullback metric.  The ambient metric $\widetilde g$ itself is not assumed
to be differentiable.

There is no loss in requiring the extension to be time-oriented.  If an
ambient extension is not time-orientable, let
$\pi:\widehat M\to\widetilde M$ be its time-orientation double cover and equip
$\widehat M$ with $\pi^*\widetilde g$.  Pulling this cover back by $\iota$
gives the time-orientation cover of $M$.  The prescribed time orientation of
$M$ selects one of its two global sections and hence a lift
$\widehat\iota:M\to\widehat M$ satisfying
$\pi\circ\widehat\iota=\iota$.  This lift is a $C^1$ isometric embedding and
its image is open.  Let $\widehat M_0$ be the connected component containing
$\widehat\iota(M)$.  Every component of a covering space over the connected
base $\widetilde M$ maps onto $\widetilde M$.  If
$\widehat\iota(M)=\widehat M_0$, then
$\iota(M)=\pi(\widehat M_0)=\widetilde M$, contrary to the properness of the
original extension.  Thus
$\widehat\iota:M\to(\widehat M_0,\pi^*\widetilde g)$ is a proper
time-oriented extension of the kind in Definition~\ref{def:C0-spacetime}.

A causal curve for a continuous metric will mean a locally Lipschitz curve
whose tangent is nonzero, future causal almost everywhere.  A timelike curve
is piecewise $C^1$, with future timelike one-sided tangents at its break
points.  Its Lorentzian length is
\[
 L_g(\gamma)=\int\sqrt{-g(\dot\gamma,\dot\gamma)}\,ds.
\]
Background on causality for continuous metrics may be found in
\cite{ChruscielGrant,Saemann}.

\begin{definition}[Future and past boundary]\label{def:boundary-plus}
Let $\iota:M\to\widetilde M$ be a $C^0$ extension.  A point
$p\in\partial\iota(M)$ belongs to the \emph{future boundary}
$\partial^+\iota(M)$ if there is a future-directed timelike curve
$\widetilde\gamma:[0,1]\to\widetilde M$ such that
\[
 \widetilde\gamma([0,1))\subset\iota(M),
 \qquad
 \widetilde\gamma(1)=p.
\]
The past boundary $\partial^-\iota(M)$ is defined by reversing the time
orientation.
\end{definition}

The following elementary consequence of metric continuity will be used
repeatedly.  We choose its constants once and for all for the later cone and
length estimates.

\begin{lemma}[A coordinate box with uniform cone bounds]\label{lem:chart}
Let $(\widetilde M,\widetilde g)$ be a time-oriented $C^0$ spacetime of
dimension $N$, and let $p\in\widetilde M$.  There are $\rho_0>0$ and a chart
\[
 \varphi=(x^0,\ldots,x^{N-1}):\widetilde U\longrightarrow
 R_{\rho_0}:=(-\rho_0,\rho_0)^N,
 \qquad \varphi(p)=0,
\]
with the following properties.  Write
$X=X^0\partial_0+\mathbf X$ and use the Euclidean norm of the coordinate
components.
\begin{enumerate}[label=\textup{(C\arabic*)},leftmargin=2.7em]
\item\label{chart-C1}
Every nonzero $\widetilde g$-causal vector satisfies $X^0\ne0$.  It is
future-directed precisely when $X^0>0$, and in that case
\[
 |\mathbf X|\le\frac32X^0.
\]
\item\label{chart-C2}
If $X^0>0$ and $|\mathbf X|\le\frac12X^0$, then $X$ is future-directed
timelike.
\item\label{chart-C3}
Every future-directed causal vector satisfies
\[
 \sqrt{-\widetilde g(X,X)}\le\frac32X^0.
\]
\end{enumerate}
Consequently, if $0<\rho<\rho_0$ and $\sigma$ is a future-directed causal
curve contained in $\varphi^{-1}(\overline{R_\rho})$, then $x^0\circ\sigma$
is strictly increasing and
\begin{equation}\label{eq:controlled-box-length}
 L_{\widetilde g}(\sigma)
 \le \frac32\bigl(x^0(\sigma_{\rm end})-x^0(\sigma_{\rm start})\bigr)
 \le 3\rho.
\end{equation}
In particular, every such coordinate box has finite timelike diameter.
\end{lemma}

\begin{proof}
Start with a smooth coordinate chart $\psi:U_0\to O\subset\mathbb R^N$ about
$p$.  After a linear change of coordinates, the matrix of $\widetilde g_p$ is
$\eta=\operatorname{diag}(-1,1,\ldots,1)$, and $\partial_0|_p$ is
future-directed.  Choose $\rho_0>0$ so that
$\overline{R_{2\rho_0}}\subset O$ and so that on
$\overline{R_{2\rho_0}}$
\begin{equation}\label{eq:metric-entry-bound}
 |\widetilde g_{\mu\nu}-\eta_{\mu\nu}|\le\frac1{4N}
 \qquad(0\le\mu,\nu\le N-1).
\end{equation}
For the statement of the lemma, take
$\widetilde U=\psi^{-1}(R_{\rho_0})$ and
$\varphi=\psi|_{\widetilde U}$.
For coordinate vectors $X,Y$, Cauchy-Schwarz gives
\begin{equation}\label{eq:metric-bilinear-bound}
 |(\widetilde g-\eta)(X,Y)|
 \le\frac1{4N}\Bigl(\sum_\mu|X^\mu|\Bigr)
                     \Bigl(\sum_\nu|Y^\nu|\Bigr)
 \le\frac14|X|\,|Y|.
\end{equation}

If $X$ is causal, then
\[
 0\ge\widetilde g(X,X)
 \ge |\mathbf X|^2-(X^0)^2-\frac14\bigl((X^0)^2+|\mathbf X|^2\bigr).
\]
It follows that
\[
 \frac34|\mathbf X|^2\le\frac54(X^0)^2,
 \qquad
 |\mathbf X|\le\sqrt{\frac53}\,|X^0|<\frac32|X^0|.
\]
In particular, a nonzero causal vector has $X^0\ne0$.  Conversely, if
$X^0>0$ and $|\mathbf X|\le X^0/2$, then
\[
 \widetilde g(X,X)
 \le -\frac34(X^0)^2+\frac54|\mathbf X|^2
 \le -\frac7{16}(X^0)^2<0.
\]
Furthermore, for any causal $X$,
\[
 -\widetilde g(X,X)
 \le \frac54(X^0)^2-\frac34|\mathbf X|^2
 \le\frac54(X^0)^2.
\]
This already gives the slightly weaker estimate stated in \ref{chart-C3}.

The vector field $\partial_0$ is timelike throughout the box by
\eqref{eq:metric-entry-bound}.  Since the box is connected and
$\partial_0|_p$ is future-directed, $\partial_0$ is future-directed everywhere.
If $X$ is causal and $X^0>0$, then the preceding cone bound and
\eqref{eq:metric-bilinear-bound} give
\[
 \widetilde g(X,\partial_0)
 \le -X^0+\frac14|X|
 \le -X^0+\frac14\bigl(X^0+|\mathbf X|\bigr)
 \le-\frac38X^0<0.
\]
Thus $X$ belongs to the same time cone as $\partial_0$.  Applying the same
argument to $-X$ proves the converse in \ref{chart-C1} and also settles the
orientation assertion in \ref{chart-C2}.

Finally, along a future causal curve, $\dot x^0>0$ almost everywhere.
Integrating \ref{chart-C3} yields
\[
 L_{\widetilde g}(\sigma)
 \le\frac32\int\dot x^0\,ds
 =\frac32\bigl(x^0(\sigma_{\rm end})-x^0(\sigma_{\rm start})\bigr),
\]
and \eqref{eq:controlled-box-length} follows from
$|x^0|\le\rho$.
\end{proof}

\subsection{A timelike curve ending at the extension boundary}

\begin{lemma}[First exit]\label{lem:first-exit}
Let $\sigma:[0,1]\to\widetilde M$ be future-directed timelike, with
$\sigma(0)\in\iota(M)$ and $\sigma(1)\notin\iota(M)$.  Then there is
$a\in(0,1]$ such that
\[
 \sigma([0,a))\subset\iota(M),
 \qquad
 \sigma(a)\in\partial^+\iota(M).
\]
The time-reversed statement also holds.
\end{lemma}

\begin{proof}
The set $A=\{s\in[0,1]:\sigma(s)\notin\iota(M)\}$ is nonempty and closed.
Let $a=\min A$.  Openness of $\iota(M)$ and the assumption on $\sigma(0)$
give $a>0$.  By the definition of $a$, the curve lies in $\iota(M)$ on
$[0,a)$, and continuity gives
$\sigma(a)\in\overline{\iota(M)}\setminus\iota(M)$.  The restricted curve is
the one required in Definition~\ref{def:boundary-plus}.
\end{proof}

\begin{proposition}[Existence of a boundary-ending timelike curve]
\label{thm:boundary-curve}
For every $C^0$ extension,
\[
 \partial^+\iota(M)\cup\partial^-\iota(M)\ne\varnothing.
\]
Thus, after reversing the time orientations of both spacetimes if necessary,
there is a future-directed timelike curve in $\widetilde M$ whose restriction
to $[0,1)$ lies in $\iota(M)$ and whose final point lies in
$\partial\iota(M)$.
\end{proposition}

\begin{proof}
Assume, to the contrary, that both the future and past boundaries are empty.
Lemma \ref{lem:first-exit}, applied in both time directions, then gives a
no-exit property: every timelike curve that starts in $\iota(M)$ remains
there.

Choose $p\in\partial\iota(M)$ and a chart as in Lemma~\ref{lem:chart}, with
$\varphi(p)=0$.  Write $\pi(x^0,\mathbf x)=\mathbf x$ and define
\[
 B:=\pi\bigl(\varphi(\iota(M)\cap\widetilde U)\bigr)
 \subset(-\rho_0,\rho_0)^{N-1}.
\]
The set $B$ is open.  The vertical lines in the chart are timelike, and the
no-exit property applied in both orientations shows that
\begin{equation}\label{eq:vertical-saturation}
 \varphi(\iota(M)\cap\widetilde U)=(-\rho_0,\rho_0)\times B.
\end{equation}
Since $p\notin\iota(M)$ but $p\in\overline{\iota(M)}$, we have
$0\notin B$ and $0\in\overline B$.  Choose
\[
 \mathbf a\in B,
 \qquad
 \mathbf b\notin B,
 \qquad
 |\mathbf a|,|\mathbf b|<\rho_0/8.
\]
Consider the segment
\[
 s\longmapsto
 \varphi^{-1}\!\left(-\frac{\rho_0}{2}e_0+s\rho_0e_0
      +(1-s)\mathbf a+s\mathbf b\right),
 \qquad 0\le s\le1,
\]
whose coordinate tangent is $(\rho_0,\mathbf b-\mathbf a)$.  Moreover,
\[
 |\mathbf b-\mathbf a|<\frac{\rho_0}{4}.
\]
Hence Lemma~\ref{lem:chart}\ref{chart-C2} shows that the segment is future
timelike.
Its initial point lies in $\iota(M)$, while its final point does not, by
\eqref{eq:vertical-saturation}.  Lemma \ref{lem:first-exit} now produces a
future boundary point, a contradiction.
\end{proof}

\subsection{A finite boundary-approaching timelike geodesic}

It is useful to separate the smooth regularity statement from the
low-regularity existence theorem.

\begin{lemma}[Smooth causal maximizers]\label{lem:smooth-maximizer}
Let $(M,g)$ be smooth.  A locally length-maximizing causal curve in $M$ is,
up to an orientation-preserving reparametrization, a causal geodesic of one
fixed causal character.  If its length is positive, it is a timelike geodesic.
\end{lemma}

\begin{proof}
Here locally length-maximizing means that every parameter value has a
neighborhood on which the corresponding subarc realizes the Lorentzian
distance between its endpoints.  Let $\gamma:I\to M$ have this property and
fix $s_0\in I$.
Choose a convex normal neighborhood $V$ of $\gamma(s_0)$ and a compact
parameter interval $[a,c]\Subset I$, containing $s_0$ in its interior, so
small that $\gamma([a,c])\subset V$ and that the restricted curve is
maximizing between its endpoints.  Let $v\in T_{\gamma(a)}M$ be determined by
\[
 \exp_{\gamma(a)}v=\gamma(c).
\]
The radial geodesic $u\mapsto\exp_{\gamma(a)}(uv)$ is causal.  The Lorentzian
Gauss lemma in a convex normal neighborhood \cite[Chapters~5 and~10]{ONeill}
gives
\[
 L_g(\beta)\le L_g\bigl(u\mapsto\exp_{\gamma(a)}(uv)\bigr)
\]
for every future causal curve $\beta$ in $V$ with the same endpoints.  In the
timelike case equality holds only for an increasing reparametrization of the
radial geodesic.  In the null case, the two endpoints are joined in $V$ by a
unique radial null geodesic.  The push-up lemma shows that a causal curve with
any different image would make the endpoints chronologically related, which
is impossible for null-related points in a convex normal neighborhood.
Since $\gamma|_{[a,c]}$ is maximizing, it is an increasing reparametrization
of this radial geodesic.

The same construction works near every parameter value.  On an overlap, the
two geodesic arcs share a nontrivial subarc and therefore agree by uniqueness
for the geodesic equation.  Together they define a single unparametrized
causal geodesic.  Its causal character is constant, since
$g(\dot\gamma,\dot\gamma)$ is constant in an affine parametrization.  If that
constant vanished, every compact subarc, and hence the whole curve, would have
zero length.  Positive length therefore forces the geodesic to be timelike.
\end{proof}

\begin{lemma}[Minguzzi-Suhr boundary maximizer]
\label{lem:MS-boundary-maximizer}
Let $(M,g)$ and $(\widetilde M,\widetilde g)$ be time-oriented $C^0$
spacetimes, and let $\iota:M\to\widetilde M$ be a $C^1$
time-orientation-preserving isometric embedding with
$\partial\iota(M)\ne\varnothing$.  If
$\partial^+\iota(M)\ne\varnothing$, then for every
$q\in\partial^+\iota(M)$ and every neighborhood $U$ of $q$ in
$\widetilde M$ there is a future-directed causal curve
\[
 \sigma:[0,1]\longrightarrow U
\]
which is locally length-maximizing and satisfies
\begin{equation}\label{eq:MS-maximizer-properties}
 \sigma([0,1))\subset\iota(M),\qquad
 \sigma(1)=p\in\partial\iota(M)\cap U,\qquad
 0<L_{\widetilde g}(\sigma)<\infty.
\end{equation}
The endpoint $p$ is not required to equal the initially chosen point $q$.
\end{lemma}

\begin{proof}
For a continuous time-oriented Lorentzian metric, set
\begin{align*}
 C_x&=\{v\in T_xM\setminus\{0\}:v\text{ is future causal}\},\\
 F_x(v)&=\sqrt{-g_x(v,v)}\qquad(v\in C_x).
\end{align*}
The cones $C_x$ are closed, sharp, and convex, have nonempty interior, and
depend continuously on $x$.  The function $F:C\to[0,\infty)$ is continuous,
positively homogeneous, and concave on each cone.  Indeed, the reverse
Cauchy-Schwarz inequality for future causal vectors gives
\[
 -g_x(v,w)\geq F_x(v)F_x(w),
\]
and hence $F_x(v+w)\geq F_x(v)+F_x(w)$; positive homogeneity then gives
concavity.  Moreover, $F$ is positive in the interior and vanishes precisely
on the boundary.  Thus $(M,F)$ is a $C^0$ proper Lorentz-Finsler space in the
terminology of Minguzzi-Suhr.  The same is true of
$(\widetilde M,\widetilde F)$.  Since a $C^1$ embedding between manifolds of
equal dimension has open image, $\iota(M)$ is an open proper submanifold, and
the isometric identity preserves its cones and fundamental function.
Consequently $(\widetilde M,\widetilde F)$ is an extension of $(M,F)$ in the
sense of \cite[Section~2]{MinguzziSuhr}.  The hypotheses of
\cite[Theorem~2.3]{MinguzziSuhr} are therefore satisfied.  Applied at $q$ in
the neighborhood $U$, that theorem yields a causal maximizer of finite
positive length whose half-open image lies in $\iota(M)$ and whose endpoint
belongs to $\partial\iota(M)\cap U$.  In the terminology of that paper, such a
maximizer is a continuous causal curve that is locally length-maximizing.
\end{proof}

\begin{lemma}[Finite timelike geodesic approaching the boundary]
\label{lem:boundary-maximizer}
Let $(M,g)$ be a smooth, connected, time-oriented spacetime.  If $(M,g)$ has
a $C^0$ extension in the sense of Definition~\ref{def:C0-spacetime}, then,
after reversing both time orientations if necessary, there exist
$b\in(0,\infty)$, $p\in\partial\iota(M)$, and a future-directed unit-speed
timelike geodesic
\[
 \eta:[0,b)\longrightarrow M
\]
such that
\[
 \iota(\eta(\tau))\longrightarrow p
 \qquad(\tau\nearrow b).
\]
Moreover, $\eta$ eventually leaves every compact subset of $M$.
\end{lemma}

\begin{proof}
By Proposition~\ref{thm:boundary-curve}, after a simultaneous reversal of the
time orientations we have $\partial^+\iota(M)\ne\varnothing$.  Fix
$q\in\partial^+\iota(M)$ and a coordinate neighborhood $U$ of $q$.
Apply Lemma~\ref{lem:MS-boundary-maximizer} and denote the resulting curve by
$\sigma$ and its endpoint by $p$.

Set $\gamma=\iota^{-1}\circ\sigma|_{[0,1)}$.  Because $\iota(M)$ is open, the
local maximizing property of $\sigma$ passes to the smooth spacetime $(M,g)$.
Lemma~\ref{lem:smooth-maximizer}, together with the positivity in
\eqref{eq:MS-maximizer-properties}, shows that $\gamma$ is an unparametrized
timelike geodesic.  Its proper-time parametrization is a future-directed
unit-speed timelike geodesic $\eta:[0,b)\to M$, where
\[
 b=L_g(\gamma)=L_{\widetilde g}(\sigma)\in(0,\infty).
\]
Proper-time reparametrization preserves the limit in
\eqref{eq:MS-maximizer-properties}, hence
$\iota(\eta(\tau))\to p$ as $\tau\nearrow b$.

Suppose, finally, that $\eta$ does not eventually leave a compact set
$K\Subset M$.  Then there are $\tau_j\nearrow b$ and, after passing to a
subsequence, a point $x\in K$ such that $\eta(\tau_j)\to x$.  Continuity of
$\iota$ gives $\iota(\eta(\tau_j))\to\iota(x)$.  The same sequence converges
to $p$, and the Hausdorff property of $\widetilde M$ gives
$p=\iota(x)\in\iota(M)$, contrary to
$p\in\partial\iota(M)$.
\end{proof}

\begin{remark}\label{rem:no-GH-maximizer}
Lemma~\ref{lem:boundary-maximizer} does not require global hyperbolicity; that
hypothesis enters only in the boundary-graph statement below.  Accordingly,
the global argument may select the geodesic before classifying its possible
ends.
\end{remark}

\subsection{The boundary as a local achronal graph}

Global hyperbolicity enters the boundary argument through the following more
precise description of the future boundary.

\begin{lemma}[A graph chart at the future boundary]
\label{lem:future-boundary-graph}
Let $(M,g)$ be a smooth globally hyperbolic spacetime of dimension $N$, let
$\iota:M\to(\widetilde M,\widetilde g)$ be a $C^0$ extension, and let
$p\in\partial^+\iota(M)$.  Given $\delta>0$, there are
$\epsilon_0,\epsilon_1>0$ and a chart
\[
 \varphi:\widetilde U\longrightarrow
 (-\epsilon_0,\epsilon_0)\times
 (-\epsilon_1,\epsilon_1)^{N-1},
 \qquad \varphi(p)=0,
\]
such that
\begin{enumerate}[label=\textup{(\roman*)},leftmargin=2.3em]
\item
$|\widetilde g_{\mu\nu}-\eta_{\mu\nu}|<\delta$ throughout the chart;
\item
$(-\epsilon_0,0)\times\{0\}$ is contained in
$\varphi(\iota(M)\cap\widetilde U)$;
\item
there is a Lipschitz function
\[
 h:(-\epsilon_1,\epsilon_1)^{N-1}
 \longrightarrow(-\epsilon_0,\epsilon_0),
 \qquad h(0)=0,
\]
for which
\begin{align}
 \{(x^0,\mathbf x):-\epsilon_0<x^0<h(\mathbf x)\}
 &\subset\varphi(\iota(M)\cap\widetilde U),
 \label{eq:boundary-subgraph}\\
 \{(h(\mathbf x),\mathbf x):
       \mathbf x\in(-\epsilon_1,\epsilon_1)^{N-1}\}
 &\subset\varphi(\partial^+\iota(M)\cap\widetilde U).
 \label{eq:boundary-graph}
\end{align}
The graph in \eqref{eq:boundary-graph} is achronal relative to
$\widetilde U$.
\end{enumerate}
Consequently, a past-directed timelike curve which starts strictly below the
graph cannot cross the graph while it remains in $\widetilde U$; throughout
that part of its domain it remains in the subgraph
\eqref{eq:boundary-subgraph}, and hence in $\iota(M)$.
\end{lemma}

\begin{proof}
Let $\widetilde\gamma$ be a future-directed timelike curve in the extension
which approaches $p$ from $\iota(M)$.  On a final $C^1$ segment, straighten
$\widetilde\gamma$ to the $x^0$-axis, choose at $p$ an orthonormal basis whose
future timelike vector is tangent to that axis, and shrink the chart.  The
continuity of $\widetilde g$ gives \textup{(i)} and \textup{(ii)}, and also
makes $\partial_0$ future timelike throughout the smaller box.

We now apply \cite[Proposition~2.2]{SbierskiProof}.  Its assumptions are
exactly those available here: $(M,g)$ is globally hyperbolic,
$\iota:M\to\widetilde M$ is a $C^0$ extension, and
$p\in\partial^+\iota(M)$.  The integer denoted by $d$ in that proposition is
$N-1$ in the present notation.  The proposition supplies positive
$\epsilon_0,\epsilon_1$ and a Lipschitz height function $h$ for which
\eqref{eq:boundary-subgraph} and \eqref{eq:boundary-graph} hold and whose graph
is achronal in the coordinate rectangle.  Because $p$ has coordinate zero,
$h(0)=0$.  The proof of the cited result constructs $h$ as the upper endpoint
of the vertical columns of $\iota(I^+(S;M))$ for a Cauchy hypersurface $S$.
Global hyperbolicity prevents a column from reaching the top of the box: the
cone comparison would otherwise yield a past-inextendible timelike curve
contained in $I^+(S;M)$, contradicting the Cauchy property.  The same cone
comparison gives the Lipschitz bound, while a timelike relation between two
graph points would place the later one inside the corresponding column.  This
explains the three conclusions of the cited proposition and fixes the ambient
chronological relation used here.

For the final assertion, let $z$ lie below the graph and suppose that a
past-directed timelike curve from $z$ first meets the graph at $q$.  Put
$\widehat z=(h(\mathbf z),\mathbf z)$.  We may carry out the construction with
a smaller perturbation constant than the prescribed $\delta$; the vertical
segment from $z$ to $\widehat z$ is then future timelike.  Both endpoints have
the same spatial coordinate, and every intermediate height lies between
$z^0$ and $h(\mathbf z)$, so this segment remains inside the coordinate box
and in the closed subgraph.  The segment of the given curve from $q$ to $z$
is future timelike.  Thus
$q\ll z\ll\widehat z$, contradicting
achronality because $q$ and $\widehat z$ both lie on the graph.  Therefore the
curve stays below the graph, where
\eqref{eq:boundary-subgraph} places it in $\iota(M)$.
\end{proof}

\section{Warped spacelike singularities}\label{cen:sec:center}

\subsection{Setting and main result}

Fix a warped product of the form
\begin{equation}\label{cen:eq:general-metric}
 M_R=(0,R)_r\times\mathbb R_t\times\Sigma,
 \qquad
 g=-a(r)^2\,dr^2+b(r)^2\,dt^2+c(r)^2\gamma_\Sigma,
\end{equation}
where $(\Sigma,\gamma_\Sigma)$ is a closed connected Riemannian manifold and
$a,b,c\in C^\infty((0,R))$ are positive.  We use the time orientation for
which $r$ decreases toward the future, and write $d=\dim M_R$.  Set
\begin{equation}\label{cen:eq:AB-def}
 A(r)=\int_0^r\frac{a(s)}{b(s)}\,ds,
 \qquad
 B(r)=\int_0^r\frac{a(s)}{c(s)}\,ds.
\end{equation}
Assume that
\begin{equation}\label{cen:eq:abc-hypotheses}
 \begin{split}
 &A(R)<\infty,\qquad B(R)<\infty,\\
 &\frac{b(r)}{c(r)}\ \hbox{is nonincreasing as a function of }r,
 \qquad \lim_{r\searrow0}b(r)=+\infty.
 \end{split}
\end{equation}
In particular, $A(r),B(r)\to0$ as $r\searrow0$.  These assumptions remain
valid after decreasing $R$.

\begin{theorem}[Inextendibility at a warped spacelike singularity]
\label{cen:thm:warped-center}
Let $(\Sigma,\gamma_\Sigma)$ be a closed connected Riemannian manifold, and
let $a,b,c$ be smooth and positive on $(0,R)$.  Suppose that for some
$R_0\in(0,R)$,
\[
 \int_0^{R_0}\frac{a(r)}{b(r)}\,dr<\infty,
 \qquad
 \int_0^{R_0}\frac{a(r)}{c(r)}\,dr<\infty,
\]
the quotient $b/c$ is nonincreasing on $(0,R_0)$, and
$b(r)\to\infty$ as $r\searrow0$.  No $C^0$ extension of $(M_R,g)$ contains a
future boundary point approached by a future-directed timelike curve along
which $r\to0$.  With the opposite time orientation, no past boundary point
can be approached by a past-directed timelike curve along which $r\to0$.
\end{theorem}

The proof occupies the remainder of this section.  We first control causal
curves and timelike homotopies, then compress terminal traces away from
$r=0$, and finally localize the divergence of the radial slices in a boundary
chart.

\subsection{Causal geometry and timelike homotopies}

\begin{lemma}[Causal estimates and terminal limits]\label{cen:lem:terminal-variation}
Assume \eqref{cen:eq:abc-hypotheses}.  Every nonconstant future-directed
causal curve has strictly decreasing $r$.  If such a curve is written, with
the opposite orientation of its parameter, as
\[
 \sigma(r)=(t(r),r,\omega(r)),\qquad r\in(r_1,r_2),
\]
then, almost everywhere,
\begin{equation}\label{cen:eq:causal-budget}
 \frac{b(r)^2}{a(r)^2}|t'(r)|^2
 +\frac{c(r)^2}{a(r)^2}|\omega'(r)|_{\gamma_\Sigma}^2\leq1.
\end{equation}
The inequality is strict for a timelike curve.  Consequently, if $r_1=0$,
then $t$ and $\omega$ have limits $t_*\in\mathbb R$ and $\omega_*\in\Sigma$
at zero, and
\begin{equation}\label{cen:eq:terminal-bounds}
 |t(r)-t_*|\leq A(r),
 \qquad
 d_{\gamma_\Sigma}(\omega(r),\omega_*)\leq B(r).
\end{equation}
\end{lemma}

\begin{proof}
The covector $dr$ is timelike, since
$g^{-1}(dr,dr)=-a(r)^{-2}<0$.  Hence a nonzero causal vector cannot be tangent to a
level set of $r$.  The continuous function $X\mapsto dr(X)$ is therefore
nonzero on each selected future causal cone and has a fixed sign there.  The
chosen time orientation makes this sign negative.  Thus
$dr(\dot\sigma)<0$ almost everywhere along every future-directed causal
curve $\sigma$; since $r\circ\sigma$ is absolutely continuous, it is strictly
decreasing.

With $r$ as parameter, the causal inequality is
\[
 -a(r)^2+b(r)^2t'(r)^2
       +c(r)^2|\omega'(r)|_{\gamma_\Sigma}^2\leq0.
\]
Division by $a(r)^2$ gives \eqref{cen:eq:causal-budget}.  In
particular,
\begin{equation}\label{cen:eq:component-budgets}
 |t'(r)|\leq \frac{a(r)}{b(r)},
 \qquad
 |\omega'(r)|_{\gamma_\Sigma}\leq\frac{a(r)}{c(r)}.
\end{equation}
Both right-hand sides are integrable at zero.  Hence $t$ is Cauchy and the
fiber trace has finite length.  Completeness of the closed fiber gives the
limit of $\omega$, while integration of \eqref{cen:eq:component-budgets}
gives \eqref{cen:eq:terminal-bounds}.

\end{proof}

\begin{lemma}[Geometry of the radial level sets]\label{cen:lem:spacelike-levels}
Assume \eqref{cen:eq:abc-hypotheses}.  Every level
$S_\rho=\{r=\rho\}$ is spacelike, with complete induced metric
\begin{equation}\label{cen:eq:slice-metric}
 h_\rho=b(\rho)^2\,dt^2+c(\rho)^2\gamma_\Sigma.
\end{equation}
\end{lemma}

\begin{proof}
The restriction of $g$ to $dr=0$ is \eqref{cen:eq:slice-metric} and is
positive definite.  It is complete because it is the Riemannian product of a
constant multiple of the Euclidean line and a constant multiple of the closed
manifold $\Sigma$.
\end{proof}

\begin{lemma}[Global hyperbolicity and timelike homotopy near the singularity]\label{cen:lem:small-interior}
Let $(\Sigma,\gamma_\Sigma)$ be closed and connected, and assume
\eqref{cen:eq:abc-hypotheses}.  After decreasing $R$, the spacetime
$(M_R,g)$ has the following properties.
\begin{enumerate}[label=\textup{(\roman*)},leftmargin=2.2em]
\item $-r$ is a temporal function, each $S_\rho$ is a Cauchy hypersurface, and
      $M_R$ is globally hyperbolic;
\item any two future-directed timelike curves with the same endpoints are
      homotopic with fixed endpoints through piecewise-$C^1$
      future-directed timelike curves.
\end{enumerate}
No hypothesis on the fundamental group, orientation, or isometry group of
$\Sigma$ is required.
\end{lemma}

\begin{proof}
By Lemma~\ref{cen:lem:terminal-variation}, $r$ is strictly decreasing on
every future-directed causal curve.  Let $c$ be future-inextendible and suppose
$r(c)$ has a positive limit $r_\infty$.  On the compact radial interval
$[r_\infty,r(c(0))]$, the two right-hand sides in
\eqref{cen:eq:component-budgets} are bounded.  More explicitly, if
$s_1<s_2$ are late parameter values, the causal inequality in an arbitrary
locally Lipschitz parametrization gives
\begin{align*}
 |t(c(s_2))-t(c(s_1))|
 &\leq\int_{s_1}^{s_2}\frac{a(r)}{b(r)}|\dot r|\,ds,\\
 L_{\gamma_\Sigma}(\omega|_{[s_1,s_2]})
 &\leq\int_{s_1}^{s_2}\frac{a(r)}{c(r)}|\dot r|\,ds.
\end{align*}
Since $r$ is absolutely continuous and monotone, the one-dimensional
change-of-variables formula turns these integrals into integrals over
$[r(c(s_2)),r(c(s_1))]$.  They tend to zero as $s_1,s_2$ approach the future
end.  Thus $t$ is Cauchy and the fiber trace has a limit.  The curve converges
to a point of $M_R$ and can be continued by a short radial timelike segment,
a contradiction.  Hence $r\to0$ at the future end.  The time-reversed argument
shows that $r\to R$ at the past end.  Every causal curve which is inextendible
in both parameter directions therefore meets each $S_\rho$ exactly once.
Thus the $S_\rho$ are Cauchy hypersurfaces, and $M_R$ is globally hyperbolic.

For assertion \textup{(ii)}, let $\operatorname{conv}(\Sigma)>0$ be smaller
than a uniform strong-convexity
radius of $\Sigma$, and decrease $R$ until
\begin{equation}\label{cen:eq:B-convexity-radius}
 B(R)<\operatorname{conv}(\Sigma).
\end{equation}
Let $c_1,c_2$ be timelike curves with common endpoints.  Parametrize both by
increasing $r$, opposite to their time orientation, and write
\[
 c_i(r)=(t_i(r),r,\omega_i(r)),
 \qquad r\in[r_1,r_0],\quad i=1,2.
\]
The length of each fiber trace is less than
$B(r_0)-B(r_1)<\operatorname{conv}(\Sigma)$.  Since the two traces have the
same endpoint at $r_1$, their images lie in a single strongly convex ball.

For $u\in[r_1,r_0]$, define the accumulated length
\[
 \ell_i(u)=\int_{r_1}^{u}|\omega_i'(s)|_{\gamma_\Sigma}\,ds
\]
and, for $r\leq u$, put
\[
 \qquad
 q_i(u,r)=
 \begin{cases}
 \ell_i(r)/\ell_i(u),&\ell_i(u)>0,\\[0.4ex]
 0,&\ell_i(u)=0.
 \end{cases}
\]
The logarithm
$v_i(u)=\exp_{\omega_i(r_1)}^{-1}(\omega_i(u))$ is well defined and continuous
because the entire trace lies in the chosen strongly convex ball.  Define
\[
 \widehat\omega_{i,u}(r)
 =
 \begin{cases}
 \exp_{\omega_i(r_1)}\bigl(q_i(u,r)v_i(u)\bigr),&r\leq u,\\
 \omega_i(r),&r\geq u.
 \end{cases}
\]
At the moving splice $r=u$, the two expressions agree.  For
$\ell_i(u)>0$ and almost every $r<u$, the radial geodesic in the first line is
parametrized proportionally to arclength, and hence
\begin{align*}
 |\widehat\omega_{i,u}'(r)|_{\gamma_\Sigma}
 &=\frac{d_{\gamma_\Sigma}(\omega_i(r_1),\omega_i(u))}{\ell_i(u)}
       |\omega_i'(r)|_{\gamma_\Sigma}
 \leq |\omega_i'(r)|_{\gamma_\Sigma}.
\end{align*}
If $\ell_i(u)=0$, the original trace is constant on $[r_1,u]$, so both sides
vanish there.

We record why this construction is a genuine homotopy.  The map
$(u,r)\mapsto\widehat\omega_{i,u}(r)$ is continuous away from
$\{\ell_i(u)=0\}$ by the preceding formula and is continuous across the
moving splice because both branches converge to $\omega_i(u)$.  If
$\ell_i(u_j)\to0$, then
\[
 d_{\gamma_\Sigma}
 \bigl(\widehat\omega_{i,u_j}(r),\omega_i(r_1)\bigr)
 \leq d_{\gamma_\Sigma}(\omega_i(r_1),\omega_i(u_j))
 \leq\ell_i(u_j)
\]
for $r\leq u_j$; this gives continuity at the remaining parameter values.
For fixed $u$, the curve is piecewise $C^1$: the only new break occurs at
$r=u$, and the accumulated-length function is piecewise $C^1$ on the original
finite subdivision.  The speed estimate above shows that replacing
$\omega_i$ by $\widehat\omega_{i,u}$ while keeping $t_i$ fixed cannot increase
the left-hand side of the strict causal inequality.  Both one-sided tangents
at the new break are therefore timelike and have the same time orientation.
Thus $u\mapsto(t_i,r,\widehat\omega_{i,u})$ is a fixed-endpoint timelike
homotopy.  At $u=r_1$ it is the original curve, while at $u=r_0$ its fiber
trace runs along the unique minimizing geodesic between the common fiber
endpoints.

After applying this construction to both curves, write their fiber components
as $\zeta(\theta_i(r))$, where $\zeta:[0,\ell]\to\Sigma$ is the unit-speed
minimizing geodesic between the common endpoints.  If $\ell=0$, take
$\theta_i\equiv0$.  For $\lambda\in[0,1]$ set
\[
 t_\lambda=(1-\lambda)t_1+\lambda t_2,
 \qquad
 \theta_\lambda=(1-\lambda)\theta_1+\lambda\theta_2.
\]
The functions $\theta_i$ are piecewise $C^1$, take the values $0$ and $\ell$
at the two endpoints, and are nondecreasing.  Hence the interpolation fixes
the endpoints and remains on $\zeta$.  Convexity of $x\mapsto x^2$ gives, at
every parameter value away from the common finite subdivision,
\begin{align*}
 &-a^2+b^2(t_\lambda')^2+c^2(\theta_\lambda')^2\\
 &\quad\leq(1-\lambda)
   \{-a^2+b^2(t_1')^2+c^2(\theta_1')^2\}
 +\lambda
   \{-a^2+b^2(t_2')^2+c^2(\theta_2')^2\}<0.
\end{align*}
The same inequality holds for every one-sided tangent at a break point.
Consequently these curves form the required piecewise-$C^1$ timelike
homotopy.  Reversing the parameter restores their future orientation.
\end{proof}

\subsection{Compression of terminal traces}

The compressed traces used below need not be differentiable at their terminal
endpoint.  We therefore record a quantitative smoothing lemma.

\begin{lemma}[Smoothing uniformly timelike Lipschitz curves]
\label{cen:lem:Lipschitz-smoothing}
Let $(N,g_N)$ be a smooth time-oriented spacetime, let $q$ be an auxiliary
Riemannian metric, and let $T$ be a smooth future-directed timelike vector
field.  Suppose that $c:[a,b]\to N$ is Lipschitz, has compact image, and that
for some $\vartheta>0$ its almost-everywhere velocity satisfies
\begin{equation}\label{cen:eq:uniform-timelike-Lipschitz}
 g_N(\dot c,\dot c)\leq-\vartheta|\dot c|_q^2,
 \qquad
 g_N(T,\dot c)\leq-\vartheta|\dot c|_q,
 \qquad
 |\dot c|_q\geq\vartheta.
\end{equation}
Then $c$ can be replaced, with both endpoints fixed, by a piecewise-$C^1$
future-directed timelike curve.  The replacement may be chosen arbitrarily
close to $c$ in the compact-open topology.
\end{lemma}

\begin{proof}
Let $K=c([a,b])$.  At every differentiability point of $c$, normalize
$\dot c$ to have $q$-norm one, and let $\mathcal V$ be the closure of these
normalized vectors in $TN|_K$.  The first two inequalities in
\eqref{cen:eq:uniform-timelike-Lipschitz} are preserved under passage to the
closure.  Since the unit sphere bundle over $K$ is compact, $\mathcal V$ is a
compact subset of the future timelike cone and has positive distance, in any
bundle metric, from the null cone.

Fix $x\in K$.  Choose a coordinate ball $W_x$ with convex coordinate image.
After shrinking $W_x$, continuity of $g_N$, $T$, and $q$ gives two constant
closed convex cones in the coordinate vector space,
\begin{equation}\label{cen:eq:fixed-smoothing-cones}
 \mathcal C_x^{\rm n}\setminus\{0\}
 \subset\operatorname{Int}\mathcal C_x^{\rm w},
\end{equation}
with the following properties: every vector of $\mathcal V$ based in
$\overline{W_x}$ has its coordinate representative in
$\mathcal C_x^{\rm n}$, and every nonzero vector of
$\mathcal C_x^{\rm w}$ is future timelike for $g_N$ at every point of
$\overline{W_x}$.  To see this explicitly, choose at $x$ a linear functional
$\ell_x$ positive on the future causal cone.  The set of unit directions in
$\mathcal V$ near $x$ is compact, lies in the open half-space
$\{\ell_x>0\}$, and is uniformly timelike.  Its closed convex hull is still
contained in that half-space and in the timelike cone at $x$.  A small conical
neighborhood of this hull gives $\mathcal C_x^{\rm n}$; a slightly larger one
gives $\mathcal C_x^{\rm w}$.  Shrinking $W_x$ preserves both inclusions at
every base point.  In particular, there is $\delta_x>0$ such that
\begin{equation}\label{cen:eq:positive-smoothing-covector}
 \ell_x(v)\geq\delta_x|v|_q
 \quad\text{for every coordinate velocity of $c$ based in }W_x.
\end{equation}

Choose finitely many such balls $W_1,\ldots,W_m$ covering $K$.  A Lebesgue
number for this finite cover, together with uniform continuity of $c$, gives a
partition
\[
 a=s_0<s_1<\cdots<s_N=b
\]
so fine that each subarc $c([s_j,s_{j+1}])$ is contained in one ball $W_{k(j)}$.
In the coordinates of that ball, replace the subarc by the straight segment
joining its endpoints.  Its displacement is
\begin{equation}\label{cen:eq:polygonal-displacement}
 c(s_{j+1})-c(s_j)
 =\int_{s_j}^{s_{j+1}}\dot c(s)\,ds.
\end{equation}
The integral belongs to $\mathcal C_{k(j)}^{\rm n}$ because this cone is
closed and convex.  It is nonzero: by
\eqref{cen:eq:positive-smoothing-covector} and the lower speed bound in
\eqref{cen:eq:uniform-timelike-Lipschitz},
\[
 \ell_{k(j)}\bigl(c(s_{j+1})-c(s_j)\bigr)
 \geq\delta_{k(j)}\int_{s_j}^{s_{j+1}}|\dot c(s)|_q\,ds
 \geq\delta_{k(j)}\vartheta(s_{j+1}-s_j)>0.
\]
The coordinate segment stays in the convex ball and, by
\eqref{cen:eq:fixed-smoothing-cones}, is future timelike throughout.  The
concatenation is a piecewise-$C^1$ future-directed timelike curve with exactly
the original endpoints.  Finally, as the mesh tends to zero, uniform
continuity of $c$ and the convexity of the coordinate balls show that every
segment remains uniformly close to the subarc it replaces.  Thus the
polygonal curves converge uniformly to $c$, which proves the last assertion.
\end{proof}

Choosing the compression map in the $B$-coordinate moves the limiting spatial
trace from $r=0$ to positive radius while controlling the angular term in the
causal inequality without any symmetry of the fiber.

\begin{lemma}[Moving the terminal trace to positive radius]\label{cen:lem:trace-advancement}
Assume \eqref{cen:eq:abc-hypotheses}.  Let $\eta$ be a $C^1$
future-directed timelike curve in $M_R$ with $r\circ\eta\to0$, and write
\[
 \eta(r)=(t(r),r,\omega(r)),\qquad 0<r\leq R_0.
\]
Let $(t_*,\omega_*)$ be its terminal limit.  After decreasing $R_0$ if
necessary, for every sufficiently small $\varepsilon>0$ one has
\begin{equation}\label{cen:eq:center-point-future}
 (t_*,\varepsilon,\omega_*)\in I^+(\eta(R_0);M_R).
\end{equation}
Consequently, there are $\varepsilon_0,\delta>0$ such that
\begin{equation}\label{cen:eq:terminal-cap}
 \left\{(t,r,\omega):0<r\leq\varepsilon_0,
 |t-t_*|<\delta,
 d_{\gamma_\Sigma}(\omega,\omega_*)<\delta\right\}
 \subset I^+(\eta(R_0);M_R).
\end{equation}
There is also $\rho>0$ with the following uniform property: if
$z\ll\eta(r_s)$ for some $0<r_s<R_0$ and $r(z)<\rho$, then
\begin{equation}\label{cen:eq:terminal-uniform-past}
 z\in I^+(\eta(R_0);M_R).
\end{equation}
\end{lemma}

\begin{proof}
Fix $R_1\in(0,R_0)$.  For
$0<\varepsilon<R_1$, put
\begin{equation}\label{cen:eq:compression-constant}
 \beta_\varepsilon=\frac{B(\varepsilon)}{B(R_1)},
 \qquad
 \sigma_\varepsilon=1-\sqrt{\beta_\varepsilon}.
\end{equation}
Then $0<\sigma_\varepsilon<1$ and
$1-\sigma_\varepsilon^2=2\sqrt{\beta_\varepsilon}-\beta_\varepsilon>0$.
Define $\phi_\varepsilon:[\varepsilon,R_1]\to[0,R_1)$ by
\begin{equation}\label{cen:eq:compression-map}
 B(\phi_\varepsilon(r))
 =\sigma_\varepsilon\bigl(B(r)-B(\varepsilon)\bigr).
\end{equation}
Since $B$ is strictly increasing, $\phi_\varepsilon$ is continuous on the
closed interval and smooth away from $r=\varepsilon$; moreover,
\[
 \phi_\varepsilon(\varepsilon)=0,
 \qquad 0\leq\phi_\varepsilon(r)<r.
\]
For $r>\varepsilon$, differentiation gives
\begin{equation}\label{cen:eq:compression-derivative}
 \phi_\varepsilon'(r)
 =\sigma_\varepsilon
 \frac{a(r)c(\phi_\varepsilon(r))}
      {c(r)a(\phi_\varepsilon(r))}.
\end{equation}

Extend $t$ and $\omega$ continuously to zero.  Consider
\begin{equation}\label{cen:eq:compressed-curve}
 \widehat\eta_\varepsilon(r)
 =\bigl(t(\phi_\varepsilon(r)),r,
        \omega(\phi_\varepsilon(r))\bigr),
 \qquad \varepsilon\leq r\leq R_1.
\end{equation}
At every point at which the derivatives exist,
\begin{align}
 &\frac{c(r)}{a(r)}\,\phi_\varepsilon'(r)
       |\omega'(\phi_\varepsilon(r))|
 \notag\\
 &\hspace{4em}=\sigma_\varepsilon
       \frac{c(\phi_\varepsilon(r))}{a(\phi_\varepsilon(r))}
       |\omega'(\phi_\varepsilon(r))|,
 \label{cen:eq:angular-budget-preserved}\\
 &\frac{b(r)}{a(r)}\phi_\varepsilon'(r)
       |t'(\phi_\varepsilon(r))|
 \notag\\
 &\hspace{4em}=\sigma_\varepsilon
       \frac{(b/c)(r)}{(b/c)(\phi_\varepsilon(r))}
       \frac{b(\phi_\varepsilon(r))}{a(\phi_\varepsilon(r))}
       |t'(\phi_\varepsilon(r))|
 \notag\\
 &\hspace{4em}\leq\sigma_\varepsilon
       \frac{b(\phi_\varepsilon(r))}{a(\phi_\varepsilon(r))}
       |t'(\phi_\varepsilon(r))|.
 \label{cen:eq:time-budget-improved}
\end{align}
Here we used $\phi_\varepsilon(r)\leq r$ and the assumed monotonicity of
$b/c$.
Together with the strict form of \eqref{cen:eq:causal-budget}, these estimates
give
\begin{equation}\label{cen:eq:compressed-strict-budget}
 \begin{split}
 &\frac{b(r)^2}{a(r)^2}
       \left|\frac d{dr}t(\phi_\varepsilon(r))\right|^2
 +\frac{c(r)^2}{a(r)^2}\left|\frac d{dr}
       \omega(\phi_\varepsilon(r))\right|^2\\
 &\hspace{10em}<\sigma_\varepsilon^2<1
 \end{split}
\end{equation}
almost everywhere.

Regularity at $r=\varepsilon$ is most easily seen in the coordinate $u=B(r)$.
The fiber trace is $1$-Lipschitz in this coordinate, because its metric
derivative is
\[
 \left|\frac{d\omega}{du}\right|
 =\frac{c(r)}{a(r)}|\omega'(r)|\leq1.
\]
This metric estimate gives the required manifold-valued Lipschitz statement.
Indeed, since $\omega(r)\to\omega_*$, a terminal part of the trace lies in a
normal coordinate ball about $\omega_*$.  On a smaller concentric ball the
Euclidean and Riemannian norms are uniformly equivalent, so the coordinate
representative of $\omega\circ B^{-1}$ is Lipschitz at $u=0$.  Away from that
terminal ball the trace is $C^1$ on a compact interval and hence locally
Lipschitz; a finite subdivision gives the assertion on the whole range.
The function $t\circ B^{-1}$ is also Lipschitz near $u=0$, since
\begin{equation}\label{cen:eq:t-in-B-coordinate}
 \left|\frac d{du}(t\circ B^{-1})(u)\right|
 \leq \frac{c(r)}{b(r)}
\end{equation}
and the right-hand side is bounded on $(0,R_1]$: monotonicity of $b/c$ gives
$c(r)/b(r)\leq c(R_1)/b(R_1)$.  The
right-hand side of \eqref{cen:eq:compression-map} is a smooth Lipschitz
function of $r$ on $[\varepsilon,R_1]$.  It follows that both nonradial
components of \eqref{cen:eq:compressed-curve} are Lipschitz, including at
$r=\varepsilon$.  This proves that $\widehat\eta_\varepsilon$ is a Lipschitz
curve; no differentiability of $\phi_\varepsilon$ at its left endpoint is
needed.

Smoothing also requires a uniform cone margin.  The image of
\eqref{cen:eq:compressed-curve} is compact.  Choose
$0<\varepsilon_-<\varepsilon$ and $R_1<R_+<R$, fix an auxiliary Riemannian
metric $q$, and take a compact neighborhood of the image contained in
\[
 \{\,\varepsilon_-\leq r\leq R_+\,\}.
\]
The velocity
$v$ of the increasing-$r$ parametrization has $dr(v)=1$.  Hence there are
constants $0<m_\varepsilon\leq M_\varepsilon<\infty$ such that
\[
 m_\varepsilon\leq|v|_q\leq M_\varepsilon
\]
almost everywhere; the upper bound follows also directly from
\eqref{cen:eq:compressed-strict-budget}.  Let
$a_{\min}=\min_{[\varepsilon,R_1]}a>0$.  Then
\begin{equation}\label{cen:eq:compressed-metric-margin}
 g(v,v)\leq-a(r)^2(1-\sigma_\varepsilon^2)
 \leq-a_{\min}^2(1-\sigma_\varepsilon^2).
\end{equation}
After reversing the parameter, $-v$ is future-directed.  The closure of the
set of normalized vectors $-v/|v|_q$ is therefore a compact subset of the
future timelike cone.  For a fixed smooth future timelike field $T$, compactness
and \eqref{cen:eq:compressed-metric-margin} give a number
$\vartheta_\varepsilon>0$ for which all three inequalities in
\eqref{cen:eq:uniform-timelike-Lipschitz} hold.  Lemma
\ref{cen:lem:Lipschitz-smoothing} supplies a fixed-endpoint
piecewise-$C^1$ future timelike curve from
$(t(\phi_\varepsilon(R_1)),R_1,
\omega(\phi_\varepsilon(R_1)))$ to
$(t_*,\varepsilon,\omega_*)$.

On $[R_1,R_0]$ extend $\phi_\varepsilon$ affinely by imposing
$\phi_\varepsilon(R_0)=R_0$.  Since
$\sigma_\varepsilon\to1$ and $B(\varepsilon)\to0$, one has
\[
 \phi_\varepsilon(R_1)\to R_1,
 \qquad
 \phi_\varepsilon'|_{(R_1,R_0)}\to1.
\]
Indeed, on this interval the affine slope is
\[
 \frac{R_0-\phi_\varepsilon(R_1)}{R_0-R_1},
\]
which tends to one, and the uniform norm of
$\phi_\varepsilon-\operatorname{id}$ tends to zero.  Since $t$ and $\omega$
are $C^1$ on a compact neighborhood of $[R_1,R_0]$, the corresponding
compositions and their one-sided derivatives converge uniformly to those of
the original curve.  Thus the extended maps converge to the identity in the
piecewise-$C^1$ topology.  The original curve has a positive timelike margin
on the compact interval $[R_1,R_0]$.  Consequently, for all sufficiently
small $\varepsilon$, the curve
\[
 r\longmapsto
 \bigl(t(\phi_\varepsilon(r)),r,
       \omega(\phi_\varepsilon(r))\bigr),
 \qquad R_1\leq r\leq R_0,
\]
is timelike.  It joins $\eta(R_0)$ to the initial point of the smoothed inner
curve.  With decreasing $r$ as parameter direction, the outer piece has the
same future orientation as $\eta$, while the smoothing lemma was applied to
$-v$ and hence gives the same orientation on the inner piece.  Their
one-sided tangents at the junction lie in the future timelike cone.  The
concatenation is therefore future timelike and proves
\eqref{cen:eq:center-point-future}.

Fix one such $\varepsilon_0$.  Openness of chronological futures gives
$\delta>0$ such that
\[
 (t,\varepsilon_0,\omega)\in I^+(\eta(R_0))
\]
whenever $|t-t_*|<\delta$ and
$d_{\gamma_\Sigma}(\omega,\omega_*)<\delta$.  The radial curve with constant
$(t,\omega)$ and decreasing $r$ is future timelike.  Appending it proves
\eqref{cen:eq:terminal-cap}.

For the uniform assertion, suppose $z\ll\eta(r_s)$.  Concatenate a timelike
curve from $z$ to
$\eta(r_s)$ with the terminal part of $\eta$.  Applying
\eqref{cen:eq:terminal-bounds} to this concatenation yields
\begin{equation}\label{cen:eq:past-terminal-cap-bounds}
 |t(z)-t_*|\leq A(r(z)),
 \qquad
 d_{\gamma_\Sigma}(\omega(z),\omega_*)\leq B(r(z)).
\end{equation}
Choose $\rho\leq\varepsilon_0$ so that $A(\rho),B(\rho)<\delta$.
Then $r(z)<\rho$ places $z$ in the set on the left-hand side of
\eqref{cen:eq:terminal-cap}, proving \eqref{cen:eq:terminal-uniform-past}.
The choice of $\rho$ is independent of $r_s$.
\end{proof}

\subsection{Boundary localization}

The boundary chart will also use the following elementary Euclidean cone
geometry.
For $a\in(0,1)$ define
\[
\overline C_a^\pm
 =\{X\neq0:\ \pm X^0/|X|\geq a\}\cup\{0\},
 \qquad
 C_a^\pm=\operatorname{Int}\overline C_a^\pm,
 \qquad \kappa(a)=\sqrt{a^{-2}-1}.
\]
Thus $X\in\overline C_a^+$ exactly when
$X^0\geq0$ and $|\mathbf X|\leq\kappa(a)X^0$.

\begin{lemma}[A cone lens in the boundary chart]\label{cen:lem:cone-lens}
Suppose that, in a rectangular coordinate chart, every future causal vector
belongs to $\overline C_{a_0}^+$, where $0<a_0<a_1<1$.  Put
$\kappa_i=\kappa(a_i)$.  For $\tau>0$ set
\[
 y^-=(-\tau,0),
 \qquad
 a_\tau=\frac{\kappa_0}{\kappa_1}\tau,
\]
\[
 x^-=(-\tau-a_\tau,0),
 \qquad x^+=(a_\tau,0),
\]
and
\begin{equation}\label{cen:eq:cone-lens}
 \overline L=
 (x^-+\overline C_{a_1}^+)
 \cap(x^++\overline C_{a_1}^-),
 \qquad L=\operatorname{Int}\overline L.
\end{equation}
If $\tau$ is small enough that $\overline L$ is contained in the chart, then
\begin{equation}\label{cen:eq:fixed-diamond}
 \overline Q\Subset L,
 \qquad
 Q=I^-(0)\cap I^+(y^-),
\end{equation}
where both chronological relations are taken inside the rectangular chart and
$\overline Q$ denotes closure relative to that chart.  For the stated small
choice of $\tau$, this is also the ordinary Euclidean closure.
For $a_0=5/8$ and $a_1=6/7$, the generator vectors of the lower and upper
conical faces lie, respectively, in $\overline C_{6/7}^+$ and
$\overline C_{6/7}^-$.
\end{lemma}

\begin{proof}
Let $z=(u,\mathbf z)\in Q$.  Integration of the convex cone
$\overline C_{a_0}^+$ along timelike curves from $y^-$ to $z$ and from $z$ to
$0$ gives
\[
 -\tau<u<0,
 \qquad
 |\mathbf z|\leq\kappa_0\min\{u+\tau,-u\}
 \leq\frac{\kappa_0\tau}{2}.
\]
After decreasing $\tau$, the compact Euclidean box defined by the corresponding
non-strict inequalities lies strictly inside the coordinate rectangle.  Thus the closure
of $Q$ cannot meet the boundary of the rectangle, and relative closure agrees
with Euclidean closure.
On the other hand,
\[
 \kappa_1(u-x^{-0})\geq\kappa_1a_\tau=\kappa_0\tau,
 \qquad
 \kappa_1(x^{+0}-u)\geq\kappa_1a_\tau=\kappa_0\tau.
\]
Both cone inequalities defining $L$ therefore hold with a margin at least
$\kappa_0\tau/2$.  The same non-strict bounds hold on $\overline Q$, so
$\overline Q\subset L$.  Once $\overline L$ lies in the coordinate box,
compactness gives \eqref{cen:eq:fixed-diamond}.  The last assertion follows
from the definitions of the two conical faces and their generating rays.
\end{proof}

\subsection{Proof of the local theorem}

\begin{proof}[Proof of Theorem~\ref{cen:thm:warped-center}]
Assume that there are an isometric embedding
$\iota:M_R\to(\widetilde M,\widetilde g)$ and a singularity-reaching timelike curve
with endpoint $p\in\partial^+\iota(M_R)$.  Restrict the embedding and the
terminal part of the curve to $M_{R_0}$.  The restricted domain is connected
and its image is open and proper.  The restricted curve still converges to
$p\notin\iota(M_{R_0})$, so
$p\in\partial^+\iota(M_{R_0})$ and the restriction is again a $C^0$
extension.  Relabel $R_0$ as $R$ and decrease $R$ further so that
Lemma~\ref{cen:lem:small-interior} applies.  The same argument shows that the
terminal curve and its boundary endpoint are unaffected by this further
restriction.

Apply Lemma~\ref{lem:future-boundary-graph} at $p$.  Choose the metric
perturbation so small that throughout the resulting chart every vector in
$C_{5/6}^+$ is future timelike, every vector in $C_{5/6}^-$ is past timelike,
and every future causal vector belongs to $\overline C_{5/8}^+$.  Set
\[
 R_{\varepsilon_0,\varepsilon_1}
 :=(-\varepsilon_0,\varepsilon_0)\times
   (-\varepsilon_1,\varepsilon_1)^{d-1},
\]
and let
\[
 \varphi:\widetilde U\longrightarrow
 R_{\varepsilon_0,\varepsilon_1}
\]
be this chart and write its future boundary as $x^0=h(\mathbf x)$.  For some
$s_*>0$, the negative vertical axis satisfies
\[
 (-s_*,0)\times\{0\}\subset
 \varphi(\iota(M_R)\cap\widetilde U).
\]
Define
\begin{equation}\label{cen:eq:axis-curve}
 \eta_s=\iota^{-1}\bigl(\varphi^{-1}(s,0)\bigr),
 \qquad -s_*<s<0.
\end{equation}
This is a $C^1$ future-directed timelike curve.  It is future-inextendible in
$M_R$, because its image converges to $p$.  The Cauchy property of the
$r$-levels and the chosen time orientation imply
\begin{equation}\label{cen:eq:axis-radius}
 r(\eta_s)\longrightarrow0\qquad(s\nearrow0).
\end{equation}

\smallskip
\noindent\emph{A fixed cone lens.}
Choose $\tau>0$ so small that $(-\tau/2,0)$ lies on a terminal portion of the
axis to which Lemma~\ref{cen:lem:trace-advancement} applies, and so that
the lens of Lemma~\ref{cen:lem:cone-lens}, with
$a_0=5/8$ and $a_1=6/7$, is compactly contained in the chart.  We identify
axis points with their $x^0$-coordinates and write
\[
 s^-=-\tau,\qquad s^+=-\tfrac12\tau,
 \qquad \eta_s=\iota^{-1}(\varphi^{-1}(s,0)).
\]
Lemma~\ref{cen:lem:cone-lens} gives
\begin{equation}\label{cen:eq:Q-in-L}
 \overline Q\Subset L,
 \qquad
 Q=I^-\bigl((0,0);R_{\varepsilon_0,\varepsilon_1}\bigr)
       \cap I^+\bigl((s^-,0);R_{\varepsilon_0,\varepsilon_1}\bigr).
\end{equation}

For each $s\in(s^-,0)$ one has
\begin{align}
 &\varphi^{-1}\left(
 I^-((s,0);R_{\varepsilon_0,\varepsilon_1})
 \cap I^+((s^-,0);R_{\varepsilon_0,\varepsilon_1})\right)
 \notag\\
 &\qquad=\iota\left(I^-(\eta_s;M_R)
              \cap I^+(\eta_{s^-};M_R)\right).
 \label{cen:eq:diamond-identity}
\end{align}
For the inclusion from left to right, take $z$ in the coordinate diamond.
Choose a past-directed timelike curve $\alpha$ in the coordinate rectangle
from $(s,0)$ to $z$.  Its pullback $\varphi^{-1}\circ\alpha$ begins below the
achronal boundary graph, so Lemma~\ref{lem:future-boundary-graph} keeps it in
$\iota(M_R)$; in particular, $z$ lies below the graph.  Next choose a
past-directed timelike curve $\beta$ in the coordinate rectangle from $z$ to
$(s^-,0)$.  The same lemma
keeps $\varphi^{-1}\circ\beta$ in $\iota(M_R)$.  Thus
$\varphi^{-1}(z)$ lies in the intrinsic diamond on the right-hand side of
\eqref{cen:eq:diamond-identity}.

Conversely, let $z$ belong to the intrinsic diamond and choose a timelike
curve $c$ from $\eta_{s^-}$ to $\eta_s$ through $z$.
Lemma~\ref{cen:lem:small-interior} gives a fixed-endpoint timelike homotopy
$H:[0,1]^2\to M_R$ from the axis segment to $c$.  Set
\[
 \mathcal A=\{a\in[0,1]:\iota(H(a,[0,1]))\subset\widetilde U\}.
\]
The set $\mathcal A$ contains zero and is open: for $a\in\mathcal A$,
compactness of $H(\{a\}\times[0,1])$ and continuity of $H$ give a neighborhood
of $a$ for which the images of the corresponding curves remain in
$\widetilde U$.  Suppose now that
$a_j\in\mathcal A$ and $a_j\to a$.  Every coordinate image
$\varphi\circ\iota\circ H(a_j,\cdot)$ is contained, including its endpoints,
in the closure of
\[
 Q_s=I^+((s^-,0))\cap I^-((s,0))\subset Q.
\]
The homotopy is uniformly continuous on the compact square.  Hence
$\iota\circ H(a_j,\cdot)$ converges uniformly to
$\iota\circ H(a,\cdot)$.  The set
$\varphi^{-1}(\overline Q)$ is compact, and in particular closed in
$\widetilde M$.  Passing to the limit pointwise therefore gives
\[
 \iota(H(a,[0,1]))
 \subset\varphi^{-1}(\overline Q)
 \subset\varphi^{-1}(L)\Subset\widetilde U.
\]
It follows that $a\in\mathcal A$ and hence that $\mathcal A$ is closed.  Thus
$\mathcal A=[0,1]$, so $c$ lies in the chart.  This proves
\eqref{cen:eq:diamond-identity}.

\smallskip
\noindent\emph{A compact separator.}
Define
\begin{equation}\label{cen:eq:separator}
 \mathcal E=\left[\left(\bigcup_{s^+<s<0}I^-(\eta_s)\right)
       \cap I^+(\eta_{s^-})\right]\setminus I^+(\eta_{s^+}),
 \qquad \mathcal K=\overline{\mathcal E}^{\,M_R}.
\end{equation}
Compactness follows from explicit radial bounds.  If $z\in\mathcal E$, then
$\eta_{s^-}\ll z$, and monotonicity of $r$ gives
\[
 r(z)<r(\eta_{s^-}).
\]
Integration of the causal estimate along a curve from $\eta_{s^-}$ to $z$ gives
\begin{equation}\label{cen:eq:separator-t-bound}
 |t(z)-t(\eta_{s^-})|
 \leq A(r(\eta_{s^-}))-A(r(z))
 \leq A(r(\eta_{s^-})).
\end{equation}
The last assertion of Lemma~\ref{cen:lem:trace-advancement}, applied to the
axis with initial point $\eta_{s^+}$, gives $\rho>0$, independent of
$s\in(s^+,0)$, such that
\[
 z\ll\eta_s,\quad r(z)<\rho
 \quad\Longrightarrow\quad z\in I^+(\eta_{s^+}).
\]
Therefore every $z\in\mathcal E$ satisfies $r(z)\geq\rho$.  Together with
\eqref{cen:eq:separator-t-bound} and compactness of $\Sigma$, this places
$\mathcal E$ in
the compact cylinder
\begin{equation}\label{cen:eq:separator-cylinder}
 [\rho,r(\eta_{s^-})]\times
 [t(\eta_{s^-})-A(r(\eta_{s^-})),
  t(\eta_{s^-})+A(r(\eta_{s^-}))]\times\Sigma.
\end{equation}
Consequently $\mathcal K$ is compact.  Equation
\eqref{cen:eq:diamond-identity} gives
$\iota(\mathcal E)\subset\varphi^{-1}(Q)$.  Since
$\varphi^{-1}(\overline Q)$ is compact and hence closed in $\widetilde M$,
continuity gives
$\iota(\mathcal K)\subset\varphi^{-1}(\overline Q)$.  Equivalently,
\begin{equation}\label{cen:eq:K-in-lens}
 \varphi(\iota(\mathcal K))\subset\overline Q\subset L.
\end{equation}

The set $\mathcal K$ separates the late axis from the chronological past of
$\eta_{s^-}$.  Indeed, let a past-directed timelike curve start at $\eta_s$,
where $s>s^+$, and end at a point of
$I^-(\eta_{s^-})$.  It starts in $I^+(\eta_{s^+})$ and ends outside this set.
At its first exit let $z$ denote the exit point.  Since $M_R$ is globally
hyperbolic, it is causally simple, and therefore
$z\in J^+(\eta_{s^+})$.  The relations
\[
 \eta_{s^-}\ll\eta_{s^+}\leq z\ll\eta_s
\]
and the push-up lemma imply $\eta_{s^-}\ll z$.  By construction
$z\notin I^+(\eta_{s^+})$.  Hence $z\in\mathcal E\subset\mathcal K$.

\smallskip
\noindent\emph{Translated terminal curves.}
For $\xi\in\mathbb R$, let
\[
 T_\xi(t,r,\omega)=(t+\xi,r,\omega).
\]
This is a time-orientation-preserving isometry of \eqref{cen:eq:general-metric}.
Compactness of $\mathcal K$ and \eqref{cen:eq:K-in-lens} imply that there is
$\lambda>0$ such that
\begin{equation}\label{cen:eq:translated-K}
 \varphi(\iota(T_\xi\mathcal K))\subset L\qquad(|\xi|\leq\lambda).
\end{equation}
The entire axis segment from $x^-$ to $(s^-,0)$ lies on the negative axis and
hence, by Lemma~\ref{lem:future-boundary-graph}, in
$\varphi(\iota(M_R)\cap\widetilde U)$.  Let $q^-\in M_R$ be defined by
$\varphi(\iota(q^-))=x^-$.  Since that axis segment is future timelike,
$q^-\in I^-(\eta_{s^-})$.  Shrinking $\lambda$ gives
\begin{equation}\label{cen:eq:translated-lower-vertex}
 T_{-\xi}q^-\in I^-(\eta_{s^-})
 \qquad(|\xi|\leq\lambda).
\end{equation}
Choose $s_0\in(s^+,0)$ so close to zero that $(s_0,0)\in L$.  Since
$F(x^0,\mathbf x)=x^0-h(\mathbf x)$ is negative there, continuity allows one
to shrink $\lambda$ once more so that
\begin{equation}\label{cen:eq:tube-initialization}
 \varphi(\iota(T_\xi\eta_{s_0}))\in L,
 \qquad
 F\bigl(\varphi(\iota(T_\xi\eta_{s_0}))\bigr)<0
 \qquad(|\xi|\leq\lambda).
\end{equation}

For this choice of $\lambda$,
\begin{equation}\label{cen:eq:translated-tube}
 \varphi(\iota(T_\xi\eta_s))\in L
 \qquad(s_0\leq s<0,\ |\xi|\leq\lambda).
\end{equation}
To prove this, fix $\xi$.  As long as the translated curve remains in the
chart, it cannot meet $F=0$, since that graph consists of boundary points whereas
$T_\xi\eta_s\in M_R$.  It therefore remains below the graph.  Suppose
$\iota(T_\xi\eta_s)$ leaves $\varphi^{-1}(L)$.  Since
$\overline L\Subset R_{\varepsilon_0,\varepsilon_1}$, the curve remains in
the chart up to its first exit time.  Continuity and the initial condition
\eqref{cen:eq:tube-initialization} show that the exit point has coordinate
$z\in\partial\overline L$.  Here we used
$\overline{\operatorname{Int}\overline L}=\overline L$, which holds because
$\overline L$ is a closed convex set with nonempty interior.  The exit point
is still the image of a point of $M_R$ and
therefore satisfies $F(z)<0$.

There is a past-directed timelike path in $\partial\overline L$ from $z$ to
$x^-$.  If $z$ lies on the lower conical
face, take the generator from $z$ to $x^-$.  If $z$ lies on the upper face but
not the lower face, follow its upper generator toward the past until it reaches
the seam, and then follow the lower generator to $x^-$.  If
$T=x^{+0}-x^{-0}$, a unit upper generator has the form
\[
 x^++\rho\bigl(-a_1,\sqrt{1-a_1^2}\,e\bigr),
 \qquad |e|=1,
\]
and reaches the seam at $\rho=T/(2a_1)$.  At that point the lower cone
inequality is an equality.  Thus both pieces remain on
$\partial\overline L$, and every nonzero tangent to these generator segments
has temporal Euclidean ratio $a_1=6/7>5/6$.  They are therefore past timelike
for $\widetilde g$.  The
lower vertex cannot be a first future exit because $x^0$ strictly increases
on a future timelike curve.  The upper vertex has $F(x^+)>0$ and cannot be
reached while the curve remains below the graph.  These observations cover
all possible first-exit points.

The boundary path remains in the chart and starts below the graph.  The last
assertion of Lemma~\ref{lem:future-boundary-graph} keeps it in the subgraph,
so it pulls back to $M_R$.  Apply $T_{-\xi}$ to the pullback.  We obtain a
past-directed timelike curve from some $\eta_s$, $s>s^+$, to
$T_{-\xi}q^-\in I^-(\eta_{s^-})$.  The separating property of $\mathcal K$
forces this curve to meet $\mathcal K$.  Before applying $T_{-\xi}$, the
pullback of the boundary path must therefore meet $T_\xi\mathcal K$.
Equivalently, the coordinate boundary path meets
$\varphi(\iota(T_\xi\mathcal K))$.  This contradicts
\eqref{cen:eq:translated-K}, because that path lies in
$\partial\overline L$ and hence avoids $L$.  This proves
\eqref{cen:eq:translated-tube}.

\smallskip
\noindent\emph{The contradiction on radial level sets.}
Let $\pi_{\rm sp}(x^0,\mathbf x)=\mathbf x$ be the spatial coordinate
projection.  Choose an open convex rectangular set $B_0$ in the spatial
coordinate domain such that
\[
 \pi_{\rm sp}(\overline L)\subset B_0,
 \qquad \overline{B_0}\Subset(-\varepsilon_1,\varepsilon_1)^{d-1}.
\]
Choose
\begin{equation}\label{cen:eq:flow-bottom-height}
 -\varepsilon_0<\ell<
 \min\left\{\min_{\mathbf x\in\overline{B_0}}h(\mathbf x),
             \min_{z\in\overline L}z^0\right\}.
\end{equation}
The compact bottom set
\[
 \mathcal B=\{(\ell,\mathbf x):\mathbf x\in\overline{B_0}\}
\]
lies in the subgraph.  Its pullback is compact in $M_R$, and hence
\begin{equation}\label{cen:eq:bottom-radial-bound}
 r\geq r_b>0\qquad\hbox{on }\iota^{-1}(\varphi^{-1}(\mathcal B)).
\end{equation}

For fixed $\mathbf x\in B_0$, the vertical curve
\[
 u\longmapsto\varphi^{-1}(u,\mathbf x),
 \qquad \ell\leq u<h(\mathbf x),
\]
is future timelike, lies in $\iota(M_R)$, and is future-inextendible there,
because it converges to the boundary point
$\varphi^{-1}(h(\mathbf x),\mathbf x)$.  By the
first part of Lemma~\ref{cen:lem:small-interior}, its $r$-coordinate decreases
from its value at the bottom to zero.  In view of
\eqref{cen:eq:bottom-radial-bound}, it therefore meets every
$S_j=\{r=r_j\}$ with $0<r_j<r_b$ exactly once.  The pullback of the vertical
vector is timelike and cannot be tangent to the spacelike hypersurface $S_j$;
equivalently, the derivative of $r$ along the vertical curve is nonzero at the
intersection.  Since a $C^1$ embedding between manifolds of the same dimension
is a local $C^1$ diffeomorphism, the function
\[
 (u,\mathbf x)\longmapsto
 r\bigl(\iota^{-1}(\varphi^{-1}(u,\mathbf x))\bigr)
\]
is $C^1$ on the subgraph.  The implicit-function theorem produces a local
$C^1$ graph over the spatial variables.  The uniqueness of the vertical
intersection makes these local graphs agree on overlaps and extends them over
all of $B_0$.  Thus
\begin{equation}\label{cen:eq:slice-graph}
 \varphi(\iota(S_j))\cap
 \{(u,\mathbf x):\mathbf x\in B_0,\ \ell<u<h(\mathbf x)\}
 =\{(H_j(\mathbf x),\mathbf x):\mathbf x\in B_0\}
\end{equation}
for a $C^1$ function $H_j:B_0\to\mathbb R$.

The functions $H_j$ satisfy a Lipschitz bound independent of $j$.  On the
compact subgraph slab over $\overline{B_0}$ there are constants $c,C>0$ such
that
\begin{equation}\label{cen:eq:chart-coefficient-bounds}
 \widetilde g_{00}\leq-c<0,
 \qquad |\widetilde g_{\mu\nu}|\leq C.
\end{equation}
For each spatial coordinate $i$, the vector
\[
 X_i=(\partial_iH_j)\partial_0+\partial_i
\]
is tangent to the spacelike hypersurface $S_j$.  Hence
\[
 0<\widetilde g(X_i,X_i)
 \leq-c|\partial_iH_j|^2+2C|\partial_iH_j|+C.
\]
Solving this quadratic inequality gives
\begin{equation}\label{cen:eq:uniform-slope-bound}
 |\partial_iH_j|\leq
 C_{\rm sl}:=\frac{C+\sqrt{C^2+cC}}{c}
\end{equation}
for every $i$ and $j$.  If $h^{(j)}$ denotes the metric induced on the graph,
then \eqref{cen:eq:chart-coefficient-bounds} and
\eqref{cen:eq:uniform-slope-bound} imply, with a constant independent of $j$,
\begin{equation}\label{cen:eq:uniform-induced-upper-bound}
 h^{(j)}_{ab}v^av^b\leq C_{\rm ind}|v|^2,
 \qquad \mathbf x\in B_0.
\end{equation}
Indeed, one may take
\[
 C_{\rm ind}=(d-1)C(C_{\rm sl}^2+2C_{\rm sl}+1).
\]
Joining two spatial points by the Euclidean segment in the convex set $B_0$
and lifting this segment to the graph gives
\begin{equation}\label{cen:eq:chart-slice-distance-bound}
 d_{S_j}(p,q)
 \leq C_0:=\sqrt{C_{\rm ind}}\,
       \operatorname{diam}_{\rm Euc}(B_0)
\end{equation}
whenever the coordinate images of $p,q\in S_j$ lie in $L$ and in the
below-graph slab displayed in \eqref{cen:eq:slice-graph}.

Choose $r_j\searrow0$ and $s_j\nearrow0$ so that
$r(\eta_{s_j})=r_j$.  For all large $j$, define
\[
 p_j=T_{-\lambda}\eta_{s_j},
 \qquad q_j=T_\lambda\eta_{s_j}.
\]
Equation \eqref{cen:eq:translated-tube} places
$\varphi(\iota(p_j))$ and $\varphi(\iota(q_j))$ in $L$.  They are below the
boundary graph and above the bottom level in
\eqref{cen:eq:flow-bottom-height}; hence they lie on the graph
\eqref{cen:eq:slice-graph}.  Equation \eqref{cen:eq:chart-slice-distance-bound} gives
\begin{equation}\label{cen:eq:distance-upper-center}
 d_{S_j}(p_j,q_j)\leq C_0.
\end{equation}

By contrast, the induced metric
\eqref{cen:eq:slice-metric} gives, for every piecewise-$C^1$ curve
$\chi=(t,\omega)$ in $S_j$ from $p_j$ to $q_j$,
\[
 L_{h_{r_j}}(\chi)
 =\int\sqrt{b(r_j)^2\dot t^{\,2}
       +c(r_j)^2|\dot\omega|_{\gamma_\Sigma}^2}\,ds
 \geq b(r_j)\int|\dot t|\,ds
 \geq2\lambda b(r_j).
\]
The curve with constant fiber coordinate and monotone $t$ realizes equality.
Therefore
\begin{equation}\label{cen:eq:distance-lower-center}
 d_{S_j}(p_j,q_j)=2\lambda b(r_j)
 \longrightarrow+\infty,
\end{equation}
contradicting \eqref{cen:eq:distance-upper-center}.  This proves the future
assertion.  Reversing the time orientation proves the time-reversed statement.

\end{proof}
\section{Canonical one-horizon Birmingham-Kottler spacetimes}
\label{sec:BK-global}

We construct the canonical one-horizon warped spacetime on a single global
Kruskal domain and classify the finite proper-time ends of its timelike
geodesics.  The terminal part of each curve is first placed in one of the four
Kruskal quadrants.  The static coordinate $t$ is introduced only afterward,
once all possible horizon crossings have occurred.

Throughout the section,
\begin{equation}\label{eq:BK-global-lapse}
 f(r)=k-\frac{2m}{r^{n-2}}+\alpha^2r^2,
 \qquad
 \alpha^2=-\frac{2\Lambda}{n(n-1)},
\end{equation}
where
\begin{equation}\label{eq:BK-global-parameters}
 n\ge3,\qquad m>0,\qquad \Lambda\le0,\qquad
 k\in\{-1,0,1\},
 \qquad k=1\quad\hbox{if }\Lambda=0.
\end{equation}
Until the field equation is considered, the fiber
$(\Sigma^{n-1},\gamma_\Sigma)$ is assumed merely closed and connected.  Its
Einstein condition is imposed in Proposition~\ref{prop:BK-Einstein-equation}.

\subsection{A global Kruskal domain}

\begin{proposition}[Global Kruskal coordinates]\label{prop:global-Kruskal}
Under \eqref{eq:BK-global-parameters}, the function $f$ has precisely one
positive zero $r_h$, and this zero is simple.  Set
\begin{equation}\label{eq:BK-surface-gravity}
 \kappa:=\frac12f'(r_h)>0.
\end{equation}
There is a smooth strictly decreasing diffeomorphism
\begin{equation}\label{eq:global-G-range}
 G:(0,\infty)\longrightarrow(G_\infty,G_0),
 \qquad G(r_h)=0,
\end{equation}
where $G_0\in(0,\infty)$ and
\begin{equation}\label{eq:global-G-endpoints}
 G_\infty=
 \begin{cases}
  -\infty,&\Lambda=0,\\
  -e^{2\kappa r_*^\infty}\in(-\infty,0),&\Lambda<0
 \end{cases}
\end{equation}
for a finite constant $r_*^\infty$.  If $r=G^{-1}$ and
\begin{equation}\label{eq:global-Kruskal-domain}
 \mathcal D:=\{(U,V)\in\mathbb R^2:G_\infty<UV<G_0\},
 \qquad
 \Mcan:=\mathcal D\times\Sigma,
\end{equation}
then
\begin{equation}\label{eq:global-Kruskal-metric}
 g=-F(UV)\,dU\,dV+r(UV)^2\gamma_\Sigma,
 \qquad
 F(x):=-\frac{f(r(x))}{\kappa^2x}\quad(x\ne0),
\end{equation}
extends to a smooth Lorentzian metric on all of
$\Mcan$.  Here and below
$dU\,dV=\frac12(dU\otimes dV+dV\otimes dU)$.  More precisely, if
\begin{equation}\label{eq:tortoise-regular-part}
 r_*(r)=\frac{1}{2\kappa}\log|r-r_h|+H(r)
\end{equation}
near $r_h$, with $H$ smooth, then
\begin{equation}\label{eq:F-horizon-value}
 F(0)=\frac{2}{\kappa}e^{-2\kappa H(r_h)}>0.
\end{equation}

The vector fields
\begin{equation}\label{eq:global-time-and-Killing}
 T:=\partial_U+\partial_V,
 \qquad
 K:=\kappa(-U\partial_U+V\partial_V)
\end{equation}
are respectively timelike and Killing.  We use $T$ to time-orient the
spacetime.  Along every future-directed timelike curve both $U$ and $V$ are
strictly increasing.  Along every past-directed timelike curve both are
strictly decreasing.
\end{proposition}

\begin{proof}
Differentiating \eqref{eq:BK-global-lapse} gives
\begin{equation}\label{eq:f-prime-positive}
 f'(r)=2m(n-2)r^{-(n-1)}+2\alpha^2r>0
 \qquad(r>0).
\end{equation}
Moreover, $f(r)\to-\infty$ as $r\searrow0$.  At the other end,
$f(r)\to1$ if $\Lambda=0$ (where $k=1$), whereas
$f(r)\to+\infty$ if $\Lambda<0$.  The intermediate value theorem and
\eqref{eq:f-prime-positive} therefore give a unique positive zero $r_h$ and
$f'(r_h)>0$.

We normalize the tortoise coordinate as follows.  On a small interval about
$r_h$, the function
\begin{equation}\label{eq:regularized-tortoise-integrand}
 \chi(r):=\frac1{f(r)}-\frac{1}{2\kappa(r-r_h)}
\end{equation}
has a smooth continuation through $r_h$.  Indeed,
$f(r)=2\kappa(r-r_h)+O((r-r_h)^2)$.  Choose
\begin{equation}\label{eq:H-definition}
 H(r)=H(r_h)+\int_{r_h}^{r}\chi(s)\,ds
\end{equation}
near the root, and define $r_*$ on each component of
$(0,\infty)\setminus\{r_h\}$ by \eqref{eq:tortoise-regular-part}, continuing
it by $r_*'=f^{-1}$.  The same constant $H(r_h)$ is used on both sides.
Consequently the regular parts match to every order at $r_h$.

Define
\begin{equation}\label{eq:global-G}
 G(r):=-\operatorname{sgn}(r-r_h)e^{2\kappa r_*(r)}
 \qquad(r\ne r_h),
\end{equation}
and put $G(r_h)=0$.  In the neighborhood where
\eqref{eq:tortoise-regular-part} holds, this is the nonsingular formula
\begin{equation}\label{eq:G-smooth-formula}
 G(r)=-(r-r_h)e^{2\kappa H(r)}.
\end{equation}
It follows that $G$ is smooth at $r_h$ and
$G'(r_h)=-e^{2\kappa H(r_h)}<0$.  Away from the root,
\begin{equation}\label{eq:G-derivative}
 G'(r)=\frac{2\kappa G(r)}{f(r)}<0,
\end{equation}
because $G$ and $f$ have opposite signs.  Hence $G'$ never vanishes.

At the singular end,
\begin{equation}\label{eq:tortoise-center-asymptotic}
 \frac1{f(r)}=-\frac{1}{2m}r^{n-2}
    \bigl(1+O(r^{n-2})\bigr),
\end{equation}
so $r_*$ has a finite limit $r_*^0$ and
$G(r)\to G_0:=e^{2\kappa r_*^0}>0$.  If $\Lambda<0$, then
$f(r)=\alpha^2r^2(1+O(r^{-2}))$ as $r\to\infty$; hence $r_*$ has a finite
limit $r_*^\infty$ and $G$ has the finite negative limit in
\eqref{eq:global-G-endpoints}.  If $\Lambda=0$, then
$f^{-1}(r)=1+O(r^{-(n-2)})$ and $r_*(r)\to+\infty$ (for $n=3$ the error
integrates to $O(\log r)$, which does not affect the conclusion).  Thus
$G(r)\to-\infty$.  Together with \eqref{eq:G-derivative}, these endpoint
limits prove \eqref{eq:global-G-range} and show that the inverse $r=r(x)$ is
smooth.

Since $G_\infty<0<G_0$, the domain $\mathcal D$ is star-shaped with respect
to the origin.  Indeed, if $(U,V)\in\mathcal D$ and $0\leq s\leq1$, then
$(sU)(sV)=s^2UV$ lies between $UV$ and $0$, and hence remains in
$(G_\infty,G_0)$.  Thus $\mathcal D$ is connected.  Since $\Sigma$ is
connected as well, $\Mcan=\mathcal D\times\Sigma$ is connected.

To verify smoothness of the metric at $x=0$, put
\begin{equation}\label{eq:f-factor-at-horizon}
 q(r):=
 \begin{cases}
  f(r)/(r-r_h),&r\ne r_h,\\
  2\kappa,&r=r_h.
 \end{cases}
\end{equation}
The function $q$ is smooth and positive.  Equations
\eqref{eq:G-smooth-formula} and \eqref{eq:f-factor-at-horizon} give
\begin{equation}\label{eq:F-smooth-formula}
 -\frac{f(r)}{\kappa^2G(r)}
   =\frac{q(r)}{\kappa^2}e^{-2\kappa H(r)}.
\end{equation}
The right-hand side is smooth and positive, and at $r=r_h$ it equals
\eqref{eq:F-horizon-value}.  Since $r$ is smooth as a function of $G$, this
proves that $F$ is smooth and positive on $(G_\infty,G_0)$.

To identify \eqref{eq:global-Kruskal-metric} with the block metric away from
the horizon, put
\begin{equation}\label{eq:null-block-coordinates}
 u=t-r_*,\qquad v=t+r_*.
\end{equation}
In any one of the four quadrants choose fixed signs
$\sigma_U,\sigma_V\in\{\pm1\}$ and set
\begin{equation}\label{eq:UV-from-uv}
 U=\sigma_Ue^{-\kappa u},\qquad
 V=\sigma_Ve^{\kappa v},
 \qquad \sigma_U\sigma_V=\operatorname{sgn}G(r).
\end{equation}
Then $UV=G(r)$,
$dU=-\kappa U\,du$, and $dV=\kappa V\,dv$.  Since
\begin{equation}\label{eq:block-null-calculation}
 -f\,dt^2+f^{-1}\,dr^2=-f\,du\,dv
   =-F(UV)\,dU\,dV,
\end{equation}
the metric agrees with the Birmingham-Kottler block metric wherever
$f\ne0$.  The two quadrants with $UV<0$ give the two exterior blocks
$r>r_h$; the two quadrants with $UV>0$ give the two dynamic blocks
$0<r<r_h$.  The axes $U=0$ and $V=0$ are the entire horizon branches,
and their intersection is the bifurcation fiber $\{U=V=0\}\times\Sigma$.

The flow
\begin{equation}\label{eq:global-Killing-flow}
 \Phi_s(U,V,\omega)=(e^{-\kappa s}U,e^{\kappa s}V,\omega)
\end{equation}
is defined for every $s\in\mathbb R$, preserves $UV$, $dU\,dV$, and the
fiber metric, and has infinitesimal generator $K$.  Thus $K$ is a complete
Killing field.  Also
$g(T,T)=-F<0$, so $T$ is timelike.

Finally, let
$X=\dot U\partial_U+\dot V\partial_V+X_\Sigma$ be timelike.  From
\eqref{eq:global-Kruskal-metric},
\begin{equation}\label{eq:UV-timelike-inequality}
 F\dot U\dot V>r^2|X_\Sigma|_{\gamma_\Sigma}^2\ge0.
\end{equation}
Thus $\dot U$ and $\dot V$ have the same nonzero sign.  Moreover,
$g(T,X)=-\frac F2(\dot U+\dot V)$.  The inequality $g(T,X)<0$ defining the
future cone is therefore equivalent to $\dot U+\dot V>0$, and hence to
$\dot U>0$ and $\dot V>0$.  The past-directed statement follows by changing
both signs.
\end{proof}

\begin{lemma}[Terminal quadrant]\label{lem:bk:terminal-quadrant}
Let $c:[0,b)\to\Mcan$ be timelike, where
$0<b\le\infty$.  There is $\tau_0<b$ such that the image of
$c|_{[\tau_0,b)}$ lies in one open Kruskal quadrant.  Consequently its tail
lies in exactly one nonhorizon block, either an exterior block $r>r_h$ or a
dynamic block $0<r<r_h$.
\end{lemma}

\begin{proof}
A timelike curve has one time orientation on every connected parameter
interval.  By Proposition~\ref{prop:global-Kruskal}, each of $U\circ c$
and $V\circ c$ is strictly monotone.  Each can therefore vanish at most once.
A timelike curve cannot contain an interval in either null axis; this also
follows directly from \eqref{eq:UV-timelike-inequality}.  Choose $\tau_0$ after the (at most
two) zeros.  The signs of both $U\circ c$ and $V\circ c$ are then fixed, so
the tail lies in one open quadrant.  The last assertion follows from
$UV=G(r)$ and the sign of $G$.
\end{proof}

\subsection{The field equation and the geodesic first integrals}

\begin{proposition}[Vacuum Einstein equation]\label{prop:BK-Einstein-equation}
Assume in addition that
\begin{equation}\label{eq:Einstein-fiber-condition}
 \operatorname{Ric}_{\gamma_\Sigma}=(n-2)k\gamma_\Sigma.
\end{equation}
Then the metric \eqref{eq:global-Kruskal-metric} satisfies
\begin{equation}\label{eq:BK-Einstein-equation}
 \operatorname{Ric}_g=\frac{2\Lambda}{n-1}g
\end{equation}
on all of $\Mcan$.
\end{proposition}

\begin{proof}
We give the calculation on the open dense set $f\ne0$.  Write
\begin{equation}\label{eq:base-fiber-decomposition}
 g=h+r^2\gamma_\Sigma,
 \qquad h=-f\,dt^2+f^{-1}\,dr^2,
 \qquad d_\Sigma:=\dim\Sigma=n-1.
\end{equation}
The nonzero Christoffel symbols of $h$ in the $(t,r)$ coordinates are
\begin{equation}\label{eq:base-Christoffels}
 \Gamma^t{}_{tr}=\Gamma^t{}_{rt}=\frac{f'}{2f},
 \qquad
 \Gamma^r{}_{tt}=\frac{ff'}2,
 \qquad
 \Gamma^r{}_{rr}=-\frac{f'}{2f}.
\end{equation}
Direct substitution in the definitions of Hessian and Ricci curvature yields
\begin{equation}\label{eq:base-curvature-identities}
 |\nabla r|_h^2=f,
 \qquad \Delta_h r=f',
 \qquad \operatorname{Hess}_h r=\frac{f'}2h,
 \qquad \operatorname{Ric}_h=-\frac{f''}2h.
\end{equation}
The warped-product Ricci identities \cite[Chapter~7]{ONeill} now read
\begin{align}
 \operatorname{Ric}_g(X,Y)
 &=\operatorname{Ric}_h(X,Y)
   -d_\Sigma\,r^{-1}\operatorname{Hess}_h r(X,Y),
 &&X,Y\in T\mathcal D,\label{eq:warped-Ricci-base}\\
 \operatorname{Ric}_g(A,B)
 &=\operatorname{Ric}_{\gamma_\Sigma}(A,B)
   -\{r\Delta_h r+(d_\Sigma-1)|\nabla r|_h^2\}\gamma_\Sigma(A,B),
 &&A,B\in T\Sigma,\label{eq:warped-Ricci-fiber}\\
 \operatorname{Ric}_g(X,A)&=0.\label{eq:warped-Ricci-mixed}
\end{align}

The derivatives of $f$ are
\begin{equation}\label{eq:f-first-second-derivatives}
 f'=2m(n-2)r^{-(n-1)}+2\alpha^2r,
 \qquad
 f''=-2m(n-2)(n-1)r^{-n}+2\alpha^2.
\end{equation}
The mass terms cancel in the base equation:
\begin{equation}\label{eq:base-cancellation}
 f''+\frac{n-1}{r}f'=2n\alpha^2.
\end{equation}
Hence \eqref{eq:warped-Ricci-base} becomes
\begin{equation}\label{eq:base-Einstein-calculation}
 \operatorname{Ric}_g|_{T\mathcal D}
 =-\frac12\left(f''+\frac{n-1}{r}f'\right)h
 =-n\alpha^2h=\frac{2\Lambda}{n-1}h.
\end{equation}
On the fiber, \eqref{eq:Einstein-fiber-condition} and
\eqref{eq:warped-Ricci-fiber} give
\begin{align}
 \operatorname{Ric}_g|_{T\Sigma}
 &=\{(n-2)k-rf'-(n-2)f\}\gamma_\Sigma\notag\\
 &=\bigl\{(n-2)k
       -2m(n-2)r^{-(n-2)}-2\alpha^2r^2\notag\\
 &\hspace{4.2em}-(n-2)k
       +2m(n-2)r^{-(n-2)}-(n-2)\alpha^2r^2\bigr\}\gamma_\Sigma\notag\\
 &=-n\alpha^2r^2\gamma_\Sigma
 =\frac{2\Lambda}{n-1}r^2\gamma_\Sigma.
 \label{eq:fiber-Einstein-calculation}
\end{align}
Equations \eqref{eq:base-Einstein-calculation},
\eqref{eq:fiber-Einstein-calculation}, and
\eqref{eq:warped-Ricci-mixed} prove \eqref{eq:BK-Einstein-equation} away
from the horizon.  Both sides of \eqref{eq:BK-Einstein-equation} are
smooth tensor fields in the Kruskal coordinates, and the nonhorizon set is
dense.  The identity therefore holds on the two horizon axes and on the
bifurcation fiber as well.
\end{proof}

\begin{lemma}[Geodesic first integrals]
\label{lem:BK-first-integrals}
Let $\zeta:I\to\Mcan$ be a unit-speed timelike
geodesic.  The two quantities
\begin{equation}\label{eq:EJ-def}
 E:=-g(K,\dot\zeta),
 \qquad
 J^2:=r^4|\dot\omega|_{\gamma_\Sigma}^2
\end{equation}
are constant on $I$.  In a nonhorizon quadrant define
\begin{equation}\label{eq:t-from-Kruskal}
 t=\frac{1}{2\kappa}\log\left|\frac VU\right|;
\end{equation}
then $K=\partial_t$ and
\begin{align}
 E&=f\dot t,\label{eq:BK-energy-equation}\\
 \dot r^2&=P_{E,J}(r),
 \qquad
 P_{E,J}(r):=E^2-f(r)\left(1+\frac{J^2}{r^2}\right),
 \label{eq:BK-radial-equation}\\
 \ddot r&=\frac12P_{E,J}'(r)
 =-\frac12f'(r)\left(1+\frac{J^2}{r^2}\right)
   +\frac{f(r)J^2}{r^3}.
 \label{eq:BK-radial-acceleration}
\end{align}
The last two identities extend through the horizons.  No Killing field on
$\Sigma$ is required.
\end{lemma}

\begin{proof}
Since $K$ is Killing and $\zeta$ is affinely parametrized,
\begin{equation}\label{eq:Killing-first-integral-calculation}
 \frac d{d\tau}g(K,\dot\zeta)
 =g(\nabla_{\dot\zeta}K,\dot\zeta)
  +g(K,\nabla_{\dot\zeta}\dot\zeta)=0.
\end{equation}
Thus $E$ is constant.

The global metric is a warped product of the two-dimensional Kruskal base and
$(\Sigma,\gamma_\Sigma)$, with warping function $r$.  The vertical component
of the geodesic equation is
\begin{equation}\label{eq:fiber-geodesic-equation}
 \frac{D^{\gamma_\Sigma}\dot\omega}{d\tau}
   +2\frac{\dot r}{r}\dot\omega=0.
\end{equation}
Taking its $\gamma_\Sigma$ inner product with $\dot\omega$ gives
\begin{align}
 \frac d{d\tau}\bigl(r^4|\dot\omega|_{\gamma_\Sigma}^2\bigr)
 &=4r^3\dot r|\dot\omega|_{\gamma_\Sigma}^2
   +2r^4\left\langle
       -2\frac{\dot r}{r}\dot\omega,\dot\omega
     \right\rangle_{\gamma_\Sigma}=0.
 \label{eq:J-conservation-calculation}
\end{align}
This proves conservation of $J^2$ and also shows why no fiber isometry is
needed.

On $f\ne0$, the relation $K=\partial_t$ gives
$E=-g(\partial_t,\dot\zeta)=f\dot t$.  The unit-speed identity is
\begin{equation}\label{eq:unit-timelike-block}
 -1=-f\dot t^2+f^{-1}\dot r^2+r^2|\dot\omega|_{\gamma_\Sigma}^2.
\end{equation}
Substituting $f\dot t=E$ and
$r^4|\dot\omega|^2=J^2$ in \eqref{eq:unit-timelike-block} gives
\eqref{eq:BK-radial-equation}.

To retain the formula at turning points, we derive the acceleration without
dividing by $\dot r$.  The $r$ Euler-Lagrange equation of
\begin{equation}\label{eq:geodesic-Lagrangian}
 \mathcal L=\frac12\left(
   -f\dot t^2+f^{-1}\dot r^2+r^2|\dot\omega|_{\gamma_\Sigma}^2
 \right)
\end{equation}
is
\begin{equation}\label{eq:radial-Euler-Lagrange}
 \ddot r-\frac{f'}{2f}\dot r^2
   +\frac{ff'}2\dot t^2-fr|\dot\omega|_{\gamma_\Sigma}^2=0.
\end{equation}
Using \eqref{eq:BK-energy-equation},
\eqref{eq:BK-radial-equation}, and $|\dot\omega|^2=J^2/r^4$ in
\eqref{eq:radial-Euler-Lagrange} cancels the terms containing $E^2/f$ and
gives \eqref{eq:BK-radial-acceleration}.  The scalar $r$ is smooth on the
Kruskal spacetime, and the right-hand sides of \eqref{eq:BK-radial-equation} and
\eqref{eq:BK-radial-acceleration} are smooth for $r>0$.  A timelike curve
cannot contain a nontrivial interval in a horizon, so its nonhorizon parameter
values are dense.  Continuity extends both identities through every horizon
event.
\end{proof}

\begin{remark}[Horizon values of the first integrals]
\label{rem:bk:horizon-first-integrals}
At $r=r_h$, equation \eqref{eq:BK-radial-equation} gives
$\dot r^2=E^2$.  At an ordinary horizon point, $E\ne0$.  For example, on
$U=0$, $V\ne0$,
\begin{equation}\label{eq:E-on-U-axis}
 E=-g(\kappa V\partial_V,\dot\zeta)
   =\frac{\kappa F(0)V}{2}\dot U,
\end{equation}
and \eqref{eq:UV-timelike-inequality} implies $\dot U\ne0$.  The calculation
on $V=0$, $U\ne0$ is identical.  Consequently a horizon event with $E=0$
can occur only at the bifurcation fiber $U=V=0$.
\end{remark}

\subsection{Finite proper-time ends}

\begin{proposition}[Classification of finite timelike geodesic ends]
\label{prop:BK-geodesic-ends}
Let $\Sigma$ be closed, and let
$\zeta:[0,b)\to\Mcan$ be a unit-speed timelike
geodesic with $b<\infty$.  Exactly one of the following alternatives holds:
\begin{enumerate}[label=\textup{(\roman*)},leftmargin=2.2em]
\item\label{end-regular} $\zeta(\tau)$ converges, as $\tau\nearrow b$, to a
      point of $\Mcan$; the limiting point may lie
      in a nonhorizon block, on one horizon branch, or on the bifurcation
      fiber;
\item\label{end-center} a terminal segment lies in a dynamic quadrant and
      \begin{equation}\label{eq:finite-end-center}
       r(\zeta(\tau))\longrightarrow0
       \qquad(\tau\nearrow b).
      \end{equation}
\end{enumerate}
In particular, every geodesic with finite proper-time domain that eventually
leaves each compact subset of $\Mcan$ satisfies
\ref{end-center}; regular, horizon, and asymptotic endpoints are all excluded.
\end{proposition}

\begin{proof}
Choose either time orientation for $\zeta$; reversing it merely reverses the
monotonicities below.  Lemma~\ref{lem:bk:terminal-quadrant}
gives $\tau_0<b$ such that
$\zeta([\tau_0,b))$ lies in one open quadrant.  We use the block coordinates
$t,r$ only on this terminal interval.  Let $E$ and $J\ge0$ be the constants
of Lemma~\ref{lem:BK-first-integrals}.

Suppose first that the terminal quadrant is exterior, so $f>0$ and $r>r_h$.
The radial equation gives
\begin{equation}\label{eq:exterior-radial-Lipschitz}
 0\le\dot r^2
 =E^2-f(r)\left(1+\frac{J^2}{r^2}\right)\le E^2.
\end{equation}
In particular $E\ne0$, and $r$ is $|E|$-Lipschitz.  Because the parameter
interval has finite length, this bound makes $r$ Cauchy.  Hence
\begin{equation}\label{eq:exterior-r-limit}
 r_b:=\lim_{\tau\nearrow b}r(\tau)\in[r_h,\infty).
\end{equation}
Here $r_b$ is a finite real number, since a Cauchy function with values in
$\mathbb R$ cannot converge only in the extended real line.  Thus the bound
\eqref{eq:exterior-radial-Lipschitz} already excludes both
asymptotically flat and anti-de Sitter infinity at finite proper time, without
the use of an asymptotic coordinate.

If instead the terminal quadrant is dynamic, then $f<0$ and $0<r<r_h$.
Here
\begin{equation}\label{eq:dynamic-radial-strictness}
 \dot r^2=E^2+|f(r)|\left(1+\frac{J^2}{r^2}\right)>0.
\end{equation}
Thus $\dot r$ has a fixed sign on the connected terminal interval, $r$ is
strictly monotone, and
\begin{equation}\label{eq:dynamic-r-limit}
 r_b:=\lim_{\tau\nearrow b}r(\tau)\in[0,r_h].
\end{equation}

Suppose that
\begin{equation}\label{eq:regular-r-limit}
 r_b\in(0,r_h)\cup(r_h,\infty).
\end{equation}
Then, on a sufficiently late interval, $|f|\ge c>0$ and $r\ge c_r>0$.
Equations \eqref{eq:BK-energy-equation} and \eqref{eq:EJ-def} yield
\begin{equation}\label{eq:t-omega-bounds-regular}
 |\dot t|\le\frac{|E|}{c},
 \qquad
 |\dot\omega|_{\gamma_\Sigma}\le\frac{J}{c_r^2}.
\end{equation}
Hence $t$ is Cauchy, and the fiber trace has finite length.  Completeness of
the closed fiber gives limits $t_b\in\mathbb R$ and $\omega_b\in\Sigma$.
Therefore
$\zeta(\tau)\to(t_b,r_b,\omega_b)$ in its terminal block, which is
alternative \ref{end-regular}.

The remaining positive case is $r_b=r_h$, where the singular coordinate $t$
must be avoided at the limiting point.  Since $r\ge r_h/2$ for all
sufficiently large $\tau$, the
fiber estimate in
\eqref{eq:t-omega-bounds-regular} still shows that
$\omega(\tau)\to\omega_b\in\Sigma$.

Assume first that $E\ne0$.  Since
\begin{equation}\label{eq:P-horizon-expansion}
 P_{E,J}(r)=E^2-f(r)\left(1+\frac{J^2}{r^2}\right)
            =E^2+O(r-r_h),
\end{equation}
there is $s\in\{\pm1\}$ such that
\begin{equation}\label{eq:rdot-horizon-expansion}
 \dot r=s|E|+O(r-r_h).
\end{equation}
In particular, $|\dot r|\geq|E|/2$ on a sufficiently late interval.  The
function $r$ is therefore strictly monotone there and may be used as a
parameter.  Put
$\epsilon_E:=E/(s|E|)\in\{\pm1\}$.  Since
$f(r)=2\kappa(r-r_h)+O((r-r_h)^2)$, equations
\eqref{eq:BK-energy-equation} and
\eqref{eq:rdot-horizon-expansion} give
\begin{equation}\label{eq:regular-null-derivative}
 \frac{dt}{dr}-\frac{\epsilon_E}{f(r)}
 =\frac1{f(r)}\left(\frac E{\dot r}-\epsilon_E\right)=O(1).
\end{equation}
If $\epsilon_E=1$, then
$u=t-r_*$ has a finite limit.  Since $r_*\to-\infty$ at the horizon,
$v=u+2r_*\to-\infty$.  In the fixed quadrant formula
\eqref{eq:UV-from-uv}, $U$ therefore has a finite nonzero limit and $V\to0$.
If $\epsilon_E=-1$, then $v=t+r_*$ has a finite limit,
$u=v-2r_*\to+\infty$, and hence $V$ has a finite nonzero limit while
$U\to0$.  In both cases the limiting point belongs to one of the horizon
axes in $\Mcan$.

If $E=0$, the exterior inequality
\eqref{eq:exterior-radial-Lipschitz} is impossible; hence the approach is from
a dynamic quadrant.  Equation
\eqref{eq:BK-energy-equation} gives $\dot t=0$ throughout that quadrant, and
so
\begin{equation}\label{eq:constant-Kruskal-ratio}
 \left|\frac VU\right|=e^{2\kappa t}
\end{equation}
is a fixed positive number.  Since $UV=G(r)\to0$, equation
\eqref{eq:constant-Kruskal-ratio} implies $U\to0$ and $V\to0$.  Thus
$\zeta$ converges to $(0,0,\omega_b)$ on the bifurcation fiber.  This again is
alternative \ref{end-regular}.

The only case left by \eqref{eq:exterior-r-limit} and
\eqref{eq:dynamic-r-limit} is $r_b=0$, which is alternative
\ref{end-center}.  If \ref{end-regular} holds, choose a relatively compact
coordinate neighborhood of the limit.  A terminal segment of $\zeta$ lies in
its compact closure and therefore cannot eventually leave every compact
subset.  This proves the last assertion.
\end{proof}

The first integrals also give the proper-time rate at the central end.  The
next corollary is not needed for the inextendibility argument.

\begin{corollary}[Proper-time asymptotics at the singularity]
\label{cor:bk:center-proper-time-asymptotics}
Under alternative \ref{end-center} of
Proposition~\ref{prop:BK-geodesic-ends}, one has $\dot r<0$ on a terminal
interval.  If $J=0$, then
\begin{equation}\label{eq:center-proper-time-radial}
 b-\tau=\frac{2}{n\sqrt{2m}}r(\tau)^{n/2}
 \bigl(1+O(r(\tau)^{n-2})\bigr).
\end{equation}
If $J>0$, put $\delta=\min\{2,n-2\}$; then
\begin{equation}\label{eq:center-proper-time-angular}
 b-\tau=\frac{2}{(n+2)J\sqrt{2m}}
 r(\tau)^{(n+2)/2}\bigl(1+O(r(\tau)^\delta)\bigr).
\end{equation}
\end{corollary}

\begin{proof}
Write $D=-f$.  Near $r=0$,
\begin{equation}\label{eq:D-center-asymptotic-global}
 D(r)=2mr^{-(n-2)}\bigl(1+O(r^{n-2})\bigr).
\end{equation}
Since $r\to0$ and $r$ is decreasing, the radial equation is
\begin{equation}\label{eq:center-rdot-negative}
 \dot r=-\sqrt{E^2+D(r)\left(1+\frac{J^2}{r^2}\right)}.
\end{equation}
For $J=0$, equations \eqref{eq:D-center-asymptotic-global} and
\eqref{eq:center-rdot-negative} give
\begin{equation}\label{eq:center-P-J-zero}
 \frac1{|\dot r|}=\frac{1}{\sqrt{2m}}r^{(n-2)/2}
 \bigl(1+O(r^{n-2})\bigr).
\end{equation}
For $J>0$, the term $DJ^2/r^2$ is leading, and
\begin{equation}\label{eq:center-P-J-positive}
 \frac1{|\dot r|}=\frac{1}{J\sqrt{2m}}r^{n/2}
 \bigl(1+O(r^\delta)\bigr).
\end{equation}
In either case,
\begin{equation}\label{eq:center-time-integral}
 b-\tau=\int_0^{r(\tau)}\frac{ds}{
 \sqrt{E^2+D(s)(1+J^2/s^2)}}.
\end{equation}
Integrating \eqref{eq:center-P-J-zero} or
\eqref{eq:center-P-J-positive} proves the stated formulae.
\end{proof}

\subsection{Global \texorpdfstring{$C^0$}{C0}-inextendibility}

\begin{theorem}[\texorpdfstring{$C^0$}{C0}-inextendibility of the canonical one-horizon warped spacetime]
\label{thm:BK-C0-inextendible}
Let the parameters satisfy \eqref{eq:BK-global-parameters}, and let
$(\Sigma^{n-1},\gamma_\Sigma)$ be any closed connected Riemannian manifold.
Then the canonical spacetime
$(\Mcan,g)$ defined by
\eqref{eq:global-Kruskal-domain}-\eqref{eq:global-Kruskal-metric} is
$C^0$-inextendible.

If, in addition, \eqref{eq:Einstein-fiber-condition} holds, then this is the
canonical one-horizon vacuum Birmingham-Kottler spacetime and it satisfies
\eqref{eq:BK-Einstein-equation}.  No assumptions of roundness, homogeneity,
orientability, or simple connectivity are imposed in the
$C^0$-inextendibility statement.
\end{theorem}

\begin{proof}
Proposition~\ref{prop:global-Kruskal} constructs the smooth spacetime for
an arbitrary closed Riemannian fiber.  Suppose that it admitted a $C^0$
extension
\begin{equation}\label{eq:putative-global-extension}
 \iota:(\Mcan,g)
 \longrightarrow(\widetilde M,\widetilde g).
\end{equation}
By Lemma~\ref{lem:boundary-maximizer}, after reversing both time orientations
if necessary, there are $b<\infty$, a future-directed unit-speed timelike
geodesic $\eta:[0,b)\to\Mcan$, and
$p\in\partial\iota(\Mcan)$ such that
\begin{equation}\label{eq:boundary-geodesic-limit}
 \iota(\eta(\tau))\longrightarrow p
 \quad\hbox{as }\tau\nearrow b,
\end{equation}
and $\eta$ eventually leaves every compact subset of
$\Mcan$.

Proposition~\ref{prop:BK-geodesic-ends} gives
$r(\eta(\tau))\to0$.  By Lemma~\ref{lem:bk:terminal-quadrant}, a
terminal segment of $\eta$ lies in one dynamic Kruskal quadrant.  Equation
\eqref{eq:dynamic-radial-strictness} then shows that $\dot r<0$ on that
segment.  Choose $R\in(0,r_h)$ so small that
$k-mn r^{-(n-2)}<0$ for $0<r<R$, and then choose $\tau_R<b$ such that
$r(\eta(\tau))<R$ for $\tau\ge\tau_R$.  Within the terminal quadrant, set
\begin{equation}\label{eq:terminal-BK-interior}
 M_R=(0,R)_r\times\mathbb R_t\times\Sigma,
 \qquad
 g=-D(r)^{-1}dr^2+D(r)dt^2+r^2\gamma_\Sigma,
\end{equation}
where
\begin{equation}\label{eq:D-for-global-center}
 D(r)=-f(r)=\frac{2m}{r^{n-2}}-k-\alpha^2r^2>0,
\end{equation}
and equip $M_R$ with the time orientation for which decreasing $r$ is
future-directed.

The coefficients in the notation of Theorem~\ref{cen:thm:warped-center} are
\begin{equation}\label{bk:eq:center-coefficients}
 a(r)=D(r)^{-1/2},\qquad b(r)=D(r)^{1/2},\qquad c(r)=r.
\end{equation}
Since
\begin{equation}\label{bk:eq:center-D-asymptotic}
 D(r)=2mr^{-(n-2)}\bigl(1+O(r^{n-2})\bigr),
\end{equation}
we have $b(r)\to\infty$ and
\begin{align}
 \int_0^r\frac{a(s)}{b(s)}\,ds
 &=\frac{r^{n-1}}{2m(n-1)}\bigl(1+O(r^{n-2})\bigr),
 \label{bk:eq:center-A-check}\\
 \int_0^r\frac{a(s)}{c(s)}\,ds
 &=\frac{2}{(n-2)\sqrt{2m}}r^{(n-2)/2}
   \bigl(1+O(r^{n-2})\bigr).
 \label{bk:eq:center-B-check}
\end{align}
Finally,
\begin{equation}\label{bk:eq:center-ratio-check}
 \left(\frac{b}{c}\right)'(r)
 =\left(\frac{\sqrt D}{r}\right)'
 =\frac{k-mn r^{-(n-2)}}{r^2\sqrt D}<0
 \qquad(0<r<R).
\end{equation}
Equations \eqref{bk:eq:center-A-check}-\eqref{bk:eq:center-ratio-check}
verify all hypotheses of the local theorem on $M_R$.

The restriction $\iota_R:=\iota|_{M_R}$ is an isometric embedding onto an
open subset of $\widetilde M$.  The manifold $M_R$ is connected, and
$\iota_R$ remains a $C^1$ time-orientation-preserving isometric embedding.
Its image is proper because the point $p$ is not in $\iota_R(M_R)$, while
\eqref{eq:boundary-geodesic-limit} and the terminal segment
$\eta|_{[\tau_R,b)}$ show that
$p\in\partial^+\iota_R(M_R)$.  Thus $\iota_R$ is itself a $C^0$ extension in
the sense of Definition~\ref{def:C0-spacetime}.  Finally, $dr$ is timelike
throughout the connected dynamic block, so its sign is constant on each time
cone.  The curve $\eta$ is future-directed there and satisfies $\dot r<0$;
hence the inherited future orientation on $M_R$ is precisely the orientation
for which decreasing $r$ is future-directed, as required in
Theorem~\ref{cen:thm:warped-center}.  Applied to
\eqref{eq:terminal-BK-interior} and \eqref{eq:D-for-global-center}, the
local theorem, Theorem~\ref{cen:thm:warped-center}, excludes precisely
such a future boundary point.  This contradiction proves
$C^0$-inextendibility.

If \eqref{eq:Einstein-fiber-condition} holds, the field equation follows from
Proposition~\ref{prop:BK-Einstein-equation}.  The inextendibility argument itself
used only closedness and connectedness of the fiber.
\end{proof}

\begin{corollary}[Tangherlini spacetime]
\label{cor:ST-internal}
For every $n\ge3$ and $m>0$, the canonical
Schwarzschild-Tangherlini spacetime is $C^0$-inextendible.
\end{corollary}

\begin{proof}
Set $\Lambda=0$, $k=1$, and
$(\Sigma,\gamma_\Sigma)=(\mathbb S^{n-1},d\Omega_{n-1}^2)$ in
Theorem~\ref{thm:BK-C0-inextendible}.
\end{proof}

\section{Conclusions and Future Work}\label{sec:conclusion}

The local theorem excludes continuous extension through a warped spacelike
singularity using causal variation, compactness of the fiber, integrability of
the longitudinal and angular causal budgets, monotonicity of the relative warp
factors, and growth of the longitudinal warp factor.  The radial compression is the step
that removes the need for a transitive isometry group on the fiber.  For the
canonical one-horizon metric, the global Kruskal coordinates and the geodesic
first integrals show that every finite timelike geodesic escaping compact
subsets reaches this singularity.  These two arguments give
$C^0$-inextendibility of the full canonical spacetime.

The roles of the assumptions are separate.  The parameter range in
\eqref{eq:BK-global-parameters} gives a single simple horizon, while
$m>0$ determines the singular asymptotic
$D(r)\sim2mr^{-(n-2)}$.  The condition $n\ge3$ makes the angular variation
integrable, and closedness of $\Sigma$ gives both completeness of finite-length
fiber traces and the compactness used in the separator argument.  The Einstein
condition \eqref{eq:Einstein-fiber-condition} identifies the metric as a
vacuum Birmingham-Kottler solution; it is not used in the geodesic-end
classification or in the local $C^0$ obstruction.

Several natural extensions require new global input.  Charged solutions and
positive cosmological constant introduce further dynamic and static blocks;
degenerate horizons no longer admit the simple Kruskal normalization used
here.  For nonsymmetric fibers, multiple horizons and noncanonical horizon
identifications also raise the possibility of repeated horizon crossings.
Controlling those crossings, and determining when they can accumulate, is a
natural next problem.  A second direction is to replace the monotonicity of
$b/c$ by a sharp quantitative condition that still permits the terminal-trace
compression.
\section*{Acknowledgements}
The author is deeply grateful to his family for their steady support and
encouragement.

\section*{Statements and Declarations}

\noindent\textbf{Funding.}
This work is funded by the Indonesian Endowment Fund for Education (LPDP), on
behalf of the Indonesian Ministry of Higher Education, Science and
Technology, and managed under the EQUITY Program (Contract No.
4298/B3/DT.03.08/2025).

\medskip
\noindent\textbf{Data availability.}
No datasets were generated or analyzed for this work.

\end{document}